\documentclass[11pt]{article}

\usepackage[utf8]{inputenc}
\usepackage{authblk}
\usepackage{amsfonts,amsmath,amssymb,mathtools,amsthm}
\usepackage{stmaryrd}
\usepackage{mathdots}
\usepackage{graphicx}
\usepackage{float}
\usepackage{xcolor}
\usepackage{todonotes}
\usepackage{comment}
\usepackage{thmtools}
\makeatletter
\renewcommand\thmt@autorefsetup{%
  \@xa\def\csname\thmt@envname autorefname\@xa\endcsname\@xa{\thmt@thmname}%
}
\makeatother
\usepackage[top=2cm,bottom=2cm,left=2cm,right=2cm,a4paper]{geometry}
\usepackage[margin=2cm]{caption}
\usepackage{graphicx}
\usepackage{tikz}
\usetikzlibrary{arrows.meta,positioning,calc,fit,backgrounds}

\definecolor{linkcolour}{rgb}{0.15,0.15,0.55}
\definecolor{urlcolour}{rgb}{0.15,0.15,0.55}
\definecolor{citecolour}{rgb}{0.15,0.15,0.55}

\usepackage[linktoc=all,backref=page]{hyperref}
\hypersetup{
  colorlinks=true,
  linkcolor=linkcolour,
  urlcolor=urlcolour,
  citecolor=citecolour,
  linktoc=all,
  hypertexnames=false,
  unicode=true,
  bookmarksnumbered=false,
  pdfmenubar=true,
  pdftoolbar=true
}

\numberwithin{equation}{section}

\theoremstyle{plain}
\newtheorem{theorem}{Theorem}[section]
\newtheorem{proposition}[theorem]{Proposition}
\newtheorem{lemma}[theorem]{Lemma}
\newtheorem{corollary}[theorem]{Corollary}

\newtheorem{definition}[theorem]{Definition}
\newtheorem{remark}[theorem]{Remark}
\newtheorem{example}[theorem]{Example}

\def\br{\begin{remark}\normalfont\small}
\def\er{\end{remark}}
\def\bt{\begin{theorem}}
\def\et{\end{theorem}}
\def\bd{\begin{definition}}
\def\ed{\end{definition}}
\def\bp{\begin{proposition}}
\def\ep{\end{proposition}}
\def\bl{\begin{lemma}}
\def\el{\end{lemma}}
\def\bc{\begin{corollary}}
\def\ec{\end{corollary}}
\def\beaq{\begin{eqnarray}}
\def\eeaq{\end{eqnarray}}

\newcommand\encadremath[1]{\vbox{\hrule\hbox{\vrule\kern8pt
\vbox{\kern8pt\hbox{$\displaystyle #1$}\kern8pt}
\kern8pt\vrule}\hrule}}

\newcommand\framefig[1]{%
\begin{figure}[bth]
\hrule\hbox{\vrule\kern8pt
\vbox{\kern8pt\vbox{\begin{center}#1\end{center}}\kern8pt}
\kern8pt\vrule}\hrule
\end{figure}}
\newcommand\figureframex[3]{%
\begin{figure}[bth]
\hrule\hbox{\vrule\kern8pt
\vbox{\kern8pt\vbox{
\begin{center}\includegraphics[width=#1cm]{#2}\end{center}
\caption{#3}}\kern8pt}
\kern8pt\vrule}\hrule
\end{figure}}
\newcommand\figureframey[3]{%
\begin{figure}[bth]
\hrule\hbox{\vrule\kern8pt
\vbox{\kern8pt\vbox{
\begin{center}\includegraphics[height=#1cm]{#2}\end{center}
\caption{#3}}\kern8pt}
\kern8pt\vrule}\hrule
\end{figure}}

\newcommand{\be}{\begin{equation}}
\newcommand{\ee}{\end{equation}}
\newcommand{\beq}{\begin{equation}}
\newcommand{\eeq}{\end{equation}}
\newcommand{\bea}{\begin{eqnarray}}
\newcommand{\eea}{\end{eqnarray}}
\newcommand{\beqq}{\begin{equation*}}
\newcommand{\eeqq}{\end{equation*}}
\newcommand{\beaa}{\begin{eqnarray*}}
\newcommand{\eeaa}{\end{eqnarray*}}

\newcommand{\tr}{\operatorname{tr}\,}

\newcommand{\diag}{\operatorname{diag}}

\newcommand\blfootnote[1]{%
  \begingroup
  \renewcommand\thefootnote{}\footnote{#1}%
  \addtocounter{footnote}{-1}%
  \endgroup}

\newcommand{\Res}{\mathop{\,\mathrm{Res}\,}}

\title{\textbf{Folding to the Bessel-Clifford system, and Equivariant Topological Recursion}}
\author{%
$_1$Mohamad Alameddine\thanks{%
Section de math\'{e}matiques, Universit\'{e} de Gen\`{e}ve,
Rue du conseil-g\'{e}n\'{e}ral 7-9,
1205 Gen\`{e}ve, Switzerland.%
}
\qquad
$_2$Alexander Hock\protect\footnotemark[1]
}
\date{\vspace{-5ex}}

\begin{document}

\maketitle

\textbf{Abstract}: We study the interplay between symmetry reductions of meromorphic connections, folding, and topological recursion. As a motivating example, we realize the Bessel–Clifford system as a folding of a symmetry-invariant subspace of the $0-$member of the Painlev\'e IV hierarchy. We introduce equivariant topological recursion for spectral curves with a Klein four-group $\mathbb Z_2\times\mathbb Z_2$ symmetry and prove that the corresponding equivariant correlators can be folded to the ordinary topological recursion correlators on the quotient spectral curve. We further establish the compatibility of this folding with the associated wave functions and quantum curves. In the Bessel–Clifford example, this yields an explicit relation between the equivariant and folded wave functions and recovers the corresponding quantum curve reduction. These results provide a framework for relating symmetry reductions of integrable systems to folding operations in topological recursion.

\blfootnote{\textit{Email Addresses:}
$_{1}$\textsf{mohamad.alameddine@unige.ch}
$_{2}$\textsf{alexander.hock@unige.ch}}

\tableofcontents

\section{Introduction}\label{SecIntro}

\paragraph{General motivation.} A widely used approach to the study of certain classes of nonlinear differential equations, particularly integrable systems, is to associate auxiliary linear systems to them. This provides access to a rich collection of geometric and analytic tools, in particular those arising from the theory of meromorphic connections. After choosing a trivialization, a meromorphic connection is represented by a linear differential system, and its deformation theory can encode the nonlinear equation under consideration. The Painlevé equations provide a fundamental class of examples of this principle: they arise as the compatibility conditions of suitable isomonodromic deformation problems for meromorphic connections. This perspective naturally raises the following question: if a transformation or structural relation between two nonlinear differential equations is known, can it be lifted to a corresponding transformation between their associated linear systems or meromorphic connections? In particular, can the transformation be understood geometrically at the level of the linearized objects from which the nonlinear equations arise? A classical family of examples is provided by the Bäcklund transformations, named after the Swedish mathematician Albert V. Bäcklund (see the classical book \cite{Hermann1976TheGO}). Their interpretation and geometric origin in the context of Painlevé equations and isomonodromic deformation theory have been extensively studied and remain an active subject of research; see, for example, Chapter 6 of \cite{FIKN2006}. Understanding such transformations at the level of meromorphic connections can reveal structures that are not apparent from the nonlinear equations alone. Beyond providing a geometric interpretation of known transformations, this viewpoint may also offer a systematic framework for discovering new correspondences between differential equations and for identifying the geometric mechanisms underlying them. 

\medskip

In this article, we investigate a different phenomenon linking two classical differential equations, the first is the Bessel equation
\begin{align}
x^2 \frac{d^2 y}{dx^2} + x\frac{dy}{dx} + \left(x^2-\alpha^2\right)y=0.
\end{align}
The Bessel equation has a regular singularity at $x=0$ and an irregular, unramified singularity at $x=\infty$ of rank $1$. The parameter $\alpha\in\mathbb{C}$ determines the local exponents at the regular singularity and, in particular, controls the corresponding local monodromy; it is traditionally referred to as the order. We consider alongside it the Bessel--Clifford equation, also known as Kummer's second equation or the confluent hypergeometric limit equation,
\begin{align}
x\frac{d^2 y}{dx^2} +a\frac{dy}{dx} =y.
\end{align}
The two equations are related by a quadratic change of variables. More precisely, if $f$ is a solution of the Bessel--Clifford equation, then
\begin{equation}
y(x)=x^{a-1}f\left(-\frac{x^2}{4}\right)
\end{equation}
is a solution of the Bessel equation with $\alpha=a-1$. Thus, the solutions of the Bessel equation can be viewed, after a suitable rescaling, as pullbacks of solutions of the Bessel--Clifford equation under the double covering
\begin{equation}
x\longmapsto -\frac{x^2}{4}.
\end{equation}
This relation at the level of scalar differential equations naturally raises the question of whether an analogous correspondence exists at the level of the associated rank-two systems. More precisely, we ask whether the quadratic transformation relating the two scalar equations can be lifted to a geometric relation between the corresponding rank-two meromorphic connections, and how the singularity structures and associated formal data are transformed under such a lift. We will introduce an operation motivated by a known Painlev\'e example, the symmetric Painlev\'e II Lax representation (or Flaschka-Newell Painlev\'e II) and give a geometric origin for this phenomenon. We call this phenomenon in our case a \emph{Folding} since it allows us to pass from the Bessel system to a scaled Bessel-Clifford system.

\medskip

The quadratic transformation is therefore not merely a change of variables at the level of scalar equations. At the level of meromorphic connections it is realized by a genuine quotient construction: the involution on the spectral parameter equips the connection with a $\mathbb{Z}_2-$equivariant structure, and the quotient produces a new connection whose formal type is ramified. Thus the classical Bessel–Bessel–Clifford transformation provides a simple model for a more general phenomenon: symmetry reduction and folding can convert an unramified meromorphic connection into a ramified one. The symmetry reduction of the connection has a natural counterpart at the level of its characteristic spectral curve. This raises a second, more structural question. Topological recursion (TR) is constructed from spectral curve data and produces, among other objects, multidifferentials, free energies and wave functions. If the spectral curve carries a finite symmetry and the corresponding connection descends to a quotient, does TR admit a compatible notion of folding?

\medskip

TR is a universal formalism that defines a family of multi-differentials $\omega_{g,n}$ on a given spectral curve. While its original formulation was motivated by loop equations coming from matrix models \cite{Eynard:2007kz}, it admits applications in various fields of mathematics including Gromov-Witten theory \cite{Bouchard:2007ys,Norbury:2014,Fang2019}, intersection theory on moduli spaces of curves \cite{Eynard:2011ga,Dunin-Barkowski:2012kbi}, and integrable systems \cite{Eynard:2021sxg,Iwaki-P1,Marchal:2019bia} just to mention a few. The spectral curve for the initial data of TR is taken to be a Riemann surface $\Sigma$ along with two (usually) meromorphic functions $x$ and $y$ and a symmetric bi-differential denoted $B$. The original TR was defined near the ramified points of $x$ and required them to be simple and regular points for the function $y$ and non-vanishing for $dy$. A generalization of TR, referred to as \textit{equivariant TR} (eq-TR) is presented in \cite{Giacchetto:2021atk} developed for a different purpose and deals with spectral curves that admit a certain symmetry under a finite group action. The recursion itself accommodates this symmetry and manifests it on the recursion building the higher differentials. 

\medskip

The purpose for which eq-TR was first developed in \cite{Giacchetto:2021atk} was the study of spin Hurwitz numbers, which enumerate ramified covers of the Riemann sphere with weights determined by theta characteristics \cite{Alexandrov:2021izt}. These numbers are closely related to the Gromov--Witten theory of K\"ahler surfaces and are generated by BKP tau-functions; see also \cite{Giacchetto2025}, where this framework was related to the spin Gromov--Witten theory of $\mathbb{P}^1$ \cite{GiacchettoKramerLewanskiSauvaget2025}. On the other hand, TR has been related to meromorphic differential systems and their associated spectral curves. In this setting, the characteristic polynomial of the differential system provides, under suitable assumptions, the spectral curve serving as the initial data for the recursion, and the resulting invariants can be related to the asymptotic data of the underlying system \cite{Marchal:2019bia,Eynard:2021sxg}. Moreover, through the quantization of the spectral curve, one can construct a quantum curve whose WKB solutions are related to solutions of the original differential system. These two perspectives suggest a natural question in the presence of a symmetry of the characteristic polynomial. The symmetry of the connection induces a symmetry of this curve, if the connection can be folded to a quotient connection, and its spectral curve consequently descends to a quotient spectral curve, one may ask whether the same operation can be implemented at the level of eq-TR and TR associated with these curves. More precisely, can the passage from a symmetric spectral curve to its quotient be reflected directly in the multidifferentials produced by eq-TR and TR? And, if so, is the corresponding relation preserved after quantization, at the level of wave functions and quantum curves?

\paragraph{Main ideas and organization.} We will start from the meromorphic differential system side where we fix the divisor corresponding to the $0-$member of the Painlev\'e IV hierarchy, this corresponds to a simple finite pole and to an infinite pole of order $r_\infty=2$. We adapt the construction given in the generic case in \cite{MarchalOrantinAlameddine2022} to this, note that the lower degree at infinity requires some additional constraints, this is done in \autoref{sec2.1}. The same section then introduces the symmetry condition \autoref{DefSymmetry} and establishes the invariant locus as the Bessel system giving the Bessel equation.
The symmetry condition is an involution and has recently been used in \cite{alameddine2026symmetryreductionpainleveiv} to characterize the Flaschka-Newell Painlev\'e II hierarchy as an invariant locus of a sub-hierarchy of Painlev\'e IV. Our goal is to give two new representations of this symmetry in terms of well known operations on isomonodromy systems. The first realization is given by the well-known fact that the Flaschka-Newell Lax pair of Painlev\'e II admits a certain symmetry that allows for a certain folding operation \cite{Flaschka1980} that we define here (see for instance the book \cite{FIKN2006}). The resulting Lax matrix admits a leading order that is non-diagonalizable and thus a ramified structure at the infinite pole. We will motivate this in \autoref{sec2.2} when we revisit the folding operation and apply it in detail for the Bessel system. \autoref{TheoFolding} provides then the relation of the folding to the Bessel-Clifford system recovering the classical relation between the Bessel and Bessel-Clifford equations. In other words, the symmetry reduction of \cite{alameddine2026symmetryreductionpainleveiv} along with an additional folding operation can be seen as a twist of the divisor allowing for ramification.

\medskip

The second interpretation of the symmetry is given in \autoref{SubsecSchrod} where we discuss the relation to the work of \cite{Iwaki2019},  in particular, by passing to the Schrödinger gauge of the systems (for Bessel and Bessel-Clifford) we are able to identify the potentials of the spectral curves with the ones given in \cite{Iwaki2019} and we find that the folding operation actually provides the classical passage from the Kummer curve to the Bessel curve (we follow here the names given in \cite{Iwaki2019} and the diagram presented). Thus the symmetry reduction composed with the folding reproduces the same passage as the one presented in the confluence diagram from Kummer to Bessel curves obtained as classical limits of the confluent family of the Gauss hypergeometric equations. This provides the bridge between the meromorphic-connection construction and the TR part of the paper. This part also includes a general discussion on equivariant connections and the analysis of the symmetry in rank $2$ which breaks down to two cases. 

\medskip

We then turn to the eq-TR side. Let a spectral curve carry a finite symmetry, even when the quotient curve is well defined, it is not automatic that the TR invariants of the two curves are related in a simple way. The recursion depends not only on the functions defining the spectral curve but also on the bidifferential and on the local structure at the ramification points. Consequently, the quotient operation must be incorporated into the recursive construction itself. For the curves relevant here, the symmetry group is
\begin{align}
    G=\mathbb Z_2\times\mathbb Z_2,
\end{align}
where one involution is the hyperelliptic involution and the second is induced by the folding symmetry. We formulate the corresponding equivariant recursion and analyze the two natural choices of one-dimensional representations, or characters, entering its definition in \autoref{sec3.1} and precisely \autoref{def:equTR}. Although the two choices correspond to different equivariant inputs, we prove that they produce the same stable correlators \autoref{thm:hyperellepticequTR}. This provides a useful structural simplification and allows us to work with a single set of eq-TR differentials. The proof relies on the loop equations satisfied by the differentials along with a recursive argument. 

\medskip

We proceed to define the folding operation on symmetric curves, the goal is to relate the differentials produced by eq-TR to the differentials of the ordinary TR on the quotient spectral curve. This is done in \autoref{thm:foldingomegagn}, where it is shown that the two sets of differentials are related
    \begin{align}
        \omega^{eq}_{g,n}(z_1,\ldots,z_n)
        =2^{2g+n-2}\omega^{TR}_{g,n}(w(z_1),\ldots,w(z_n)).
    \end{align}
after the correlators are compared through the quotient map. The proof follows a similar argument used in \cite{Alexandrov:2025kzn} relating the initial data of both recursions and then establishing a recursive argument that allows for the explicit identification between the folded differentials and the differentials of the original TR. We show this on various examples covering even non-algebraic curves, among the examples the Bessel and Bessel-Clifford folding that serves as a test for the general mechanism developed throughout. After relating the set of differentials, we finally move to a relation between the quantum curves generated from eq-TR and TR. To our knowledge, this is the first discussion of the quantization procedure using eq-TR, although, our methods rely on the results relating the differentials to the ordinary TR on a quotient curve. More precisely, we define the wave function annihilated by the quantization of the spectral curve by the usual integral of the differentials. The definition being the same in the eq-TR or TR setups, we are able to conclude a relation between the quantum curves given by (\autoref{thm:quantumcurve})
 \begin{align}
        \hat{P}^{eq}_{\hbar/2}(\hat{x},\hat{y})=\hat{P}^{TR}_\hbar \bigg(\hat{x}^2,\frac{1}{2\hat{x}}\hat{y}-\frac{\hbar}{4\hat{x}^2}\bigg).
    \end{align}
Finally, by applying the construction to the same examples, we are able to recover the same folding operation between the quantum curves given in \cite{Iwaki2019}. This recovers the same folding operation between equivariant meromorphic differential systems and eq-TR folding to the ordinary TR.

\paragraph{Illustrative diagram and organization.}The main results of the paper can consequently be summarized in three steps, equivalently in \autoref{fig:symmetry-reduction-folding} where we show the different parts of the construction. 

\medskip
\begin{figure}[htbp]
    \centering
    \resizebox{\linewidth}{!}{%
    \begin{tikzpicture}[
        >=Stealth,
        node distance=9mm and 38mm,
        every node/.style={font=\small},
        block/.style={
            draw,
            rounded corners,
            align=center,
            text width=6.0cm,
            minimum height=1.0cm,
            inner sep=9pt,
            fill=white
        },
        arrow/.style={->, thick},
        fold/.style={->, thick, dashed},
        upperbox/.style={
            draw=orange!80!black,
            fill=orange!18,
            rounded corners,
            inner sep=12pt
        },
        lowerbox/.style={
            draw=green!50!black,
            fill=green!18,
            rounded corners,
            inner sep=12pt
        }
    ]

    \node[block] (piv) {
        \textbf{Equivariant connection}
    };

    \node[block, right=of piv] (bessel) {
        \textbf{Folded connection}
    };

    \node[block, below=of piv] (originalcurve) {
        \textbf{$\mathbb{Z}_2$-equivariant spectral curve}
    };

    \node[block, below=of bessel] (besselcurve) {
        \textbf{Folded spectral curve}
    };

    \node[block, below=18mm of originalcurve] (eqtr) {
        \textbf{Equivariant TR}\\
        on the equivariant spectral curve
    };

    \node[block, right=of eqtr] (ordinarytr) {
        \textbf{Ordinary TR}\\
        on the folded spectral data
    };

    \node[block, below=of eqtr] (eqquant) {
        \textbf{Equivariant wave function}\\
        Equivariant quantum curve
    };

    \node[block, right=of eqquant] (ordinaryquant) {
        \textbf{Folded wave function}\\
        Ordinary quantum curve
    };

    \begin{scope}[on background layer]
        \node[upperbox,
            fit=(piv)(bessel)(originalcurve)(besselcurve)] {};

        \node[lowerbox,
            fit=(eqtr)(ordinarytr)(eqquant)(ordinaryquant)] {};
    \end{scope}

    \draw[arrow] (piv.east) -- node[above, align=center]
        {Folding} node[below]
        {\autoref{TheoFolding}} (bessel.west);

    \draw[arrow] (piv.south) -- node[left, align=right]
        {Characteristic\\equation} (originalcurve.north);

    \draw[arrow] (bessel.south) -- node[right, align=left]
        {Characteristic\\equation} (besselcurve.north);

    \draw[fold] (originalcurve.east) -- node[above, align=center]
        {Induced symmetry\\and folding} (besselcurve.west);

    \draw[arrow] (originalcurve.south) -- (eqtr.north);

    \draw[arrow] (besselcurve.south) -- (ordinarytr.north);

    \draw[arrow] (eqtr.east) -- node[above, align=center]
        {Folding} node[below]
        {\autoref{thm:foldingomegagn}} (ordinarytr.west);

    \draw[arrow] (eqtr.south) -- node[left, align=right]
        {Wave function\\construction} (eqquant.north);

    \draw[arrow] (ordinarytr.south) -- node[right, align=left]
        {Wave function\\construction} (ordinaryquant.north);

    \draw[fold] (eqquant.east) -- node[above, align=center]
        {Quantum curve\\correspondence} node[below]
        {\autoref{thm:quantumcurve}} (ordinaryquant.west);

    \draw[arrow]
        (eqquant.west) -- ++(-1.2cm,0) |- (piv.west);

    \draw[arrow]
        (ordinaryquant.east) -- ++(1.2cm,0) |- (bessel.east);

    \end{tikzpicture}%
    }
    \caption{Symmetry reduction and folding, from meromorphic
    connections to spectral curves. This provides the initial data for
    the corresponding eq-TR and TR and their relation following the folding.}
    \label{fig:symmetry-reduction-folding}
\end{figure}
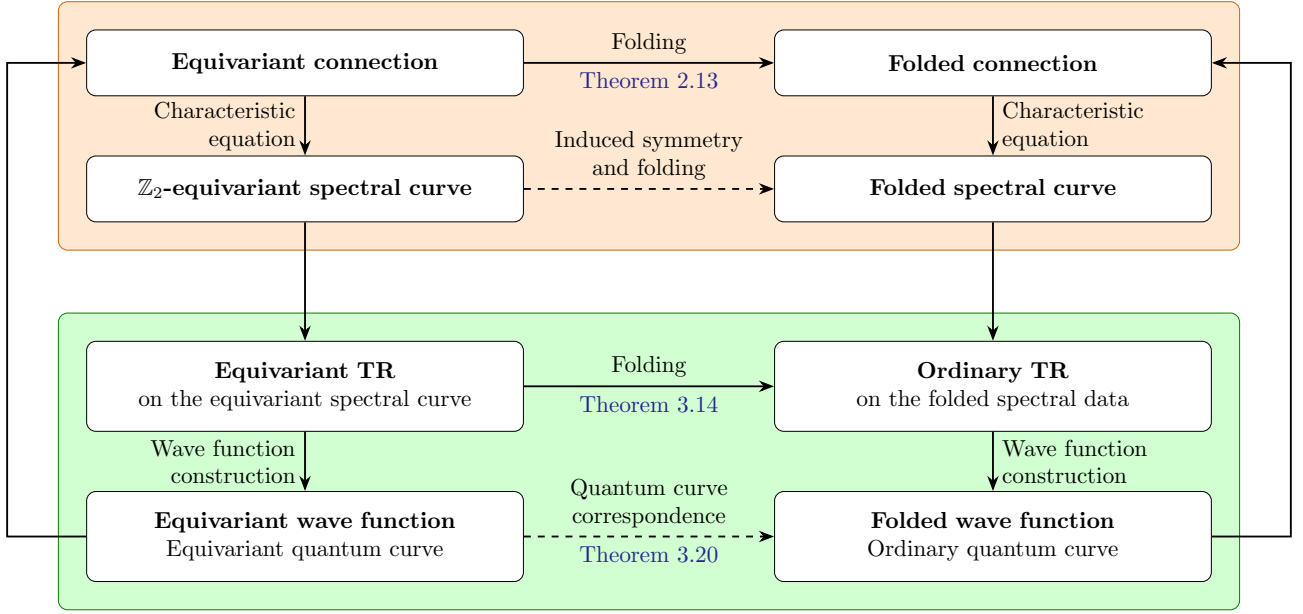

First, we identify the Bessel system as the invariant locus of a $\mathbb Z_2$ symmetry of the $0$-member of the Painlev\'e IV hierarchy and construct the corresponding folded connection. The quotient under $\lambda\mapsto-\lambda$ yields a rational connection whose scalar reduction is the Bessel--Clifford equation. This gives a geometric realization of the classical Bessel--Bessel--Clifford transformation and explains the appearance of a ramified pole on the folded side.  This is the first part of the diagram, and using the characteristic polynomial, we obtain the spectral curves of both systems as indicated. Second, we develop the corresponding eq-TR construction for hyperelliptic curves with $\mathbb Z_2\times\mathbb Z_2$ symmetry. We prove the equivalence of the two natural character choices and establish the folding theorem identifying the equivariant correlators on the covering curve with the ordinary TR correlators on the quotient. The result provides a direct correspondence between symmetry-equivariant recursion and ordinary recursion after folding. This is the first passage of the second box of the diagram.

\medskip

Third, we prove that this correspondence is preserved by the construction of wave functions and quantum curves. The resulting transformation of the quantum curves agrees with the folding obtained directly from the differential-system side. Thus the same symmetry is visible consistently at the levels of connections, spectral curves, recursive invariants, wave functions, and quantum curves. This is the final part of the diagram which states the correspondence at the level of the associated quantum curves. The Bessel--Clifford system should therefore be viewed as the simplest non-trivial model in which the complete mechanism can be made explicit. The classical special-function identity is already well known; the new point is the identification of the geometric folding behind it and the proof that this folding propagates through TR and its quantization. 

\medskip

The paper is organized as follows. In \autoref{Section2}, we construct the Bessel system as an invariant locus of the $0$-member of the Painlev\'e IV hierarchy, introduce the folding operation, and obtain the Bessel--Clifford connection. We then pass to the Schr\"odinger gauge and identify the resulting spectral curves with the curves appearing in the confluent hypergeometric setting, this is mainly the orange part of the diagram. In \autoref{sec3}, we develop the eq-TR framework for the $\mathbb Z_2\times\mathbb Z_2$-symmetric curves considered here. We prove the equivalence of the two character choices, establish the folding relation with ordinary TR, and finally prove the compatibility of the folding with wave functions and quantum curves. This is the main construction in the green box of the diagram. While several of the observations made in the present work hold in a rather general context, a lot of open questions can be raised from them. We conclude this work by \autoref{sec4} which discusses several directions suggested by the construction, including finite-group foldings, successive quotient operations, and possible applications to higher members of Painlev\'e hierarchies.

\subsection*{Acknowledgment}
The authors would like to thank Sergey Shadrin for discussions.
M. A. was partially funded by the SNSF under the grant agreement PZ00P2\_223297 and by the DFG -- project ID 551478549.  A.H. was funded by the Swiss National Science Foundation (SNSF) through the
Ambizione project “TRuality: Topological Recursion, Duality and Applications”
under the grant agreement PZ00P2\_223297 and partially by the DFG -- project ID 551478549.

\subsection*{Declaration of Generative AI and AI-assisted Technologies in the Writing Process}
Generative AI was used to assist with proofreading, identifying typographical errors, and improving the exposition. All mathematical content and scientific results are the authors' own work where no generative AI was used.

\section{The Bessel-Clifford system as a folding of an invariant locus}\label{Section2}
The starting point of this section is the genus $0$ element of the PIV hierarchy and its reduction to give the Bessel system (ultimately this leads to the Bessel equation as a quantum curve satisfied by the horizontal sections). Another goal for this section is to introduce the folding of the system giving the folded version and to recover the Bessel-Clifford system as a result. We will provide the Sch\"odinger gauge of both systems and realize the correspondence between our setup and that of \cite{Iwaki2019} giving the passage from the Kummer to Bessel curves.

\subsection{The $0-$member of the PIV hierarchy and the Bessel system} \label{sec2.1}

The purpose of this section is to establish the theory behind the $0-$element of the PIV hierarchy, for this, we will restrict our study to the case of a Lax matrix admitting a polar structure with an infinite pole $r_\infty=2$ and a tame pole that we will fix later to be at $\{ 0 \}$ in the sense of \cite{MarchalOrantinAlameddine2022}. The space of connections is defined as follows
\begin{definition}[Space of rational connections] The space of rational connections consisting of a bounded divisor by $\mathcal{D}= 2[\infty] + 1[X_0]$ is the Poisson space given by 
\begin{align}
    F_{\mathcal{D},3} := \left\{ \hat{L}(\lambda) = \hat{L}_\infty + \frac{ \hat{L}_{0}}{\lambda - X_0} \qquad
    \text{with} \qquad \left\{ \hat{L}_\infty, \hat{L}_0 \right\} \in    \mathfrak{gl}_2^2   \right\}  \big/ GL_2(\mathbb{C}), 
\end{align}
where the quotient of the complex reductive group indicates the conjugation action on all the coefficients. The horizontal sections are by definition the solutions of the equation $\nabla \Psi=0$, which after trivializing using the usual analytic structure of $\mathbb{P}^1$ becomes the meromorphic differential system 
\begin{align}
    \partial_\lambda \hat{\Psi} = \hat{L} (\lambda ) \hat{\Psi}
\end{align}
    We define the space $\hat{F}_{\mathcal{D},3} \subset F_{\mathcal{D},3}$ consisting of connections that are generic, in the sense that the eigenvalues of the leading orders at the poles are distinct, and thus both matrices are diagonalizable\footnote{We also assume non-resonance for the tame pole.}.  
\end{definition}
The number $3$ in the definition is an adaptation of the notation of \cite{MarchalOrantinAlameddine2022} and it indicates the total degree of the divisor of the connection. The Poisson structure characterizing this space is given for instance in \cite{Adams:1989ax,HarnadHurtubise} and it is inherited from the algebraic point of view of loop algebras. Following this construction, this case needs a special treatment due to the low order of the infinite pole. We will explain the necessary modifications needed to construct a global expression for the connection matrix $\hat{L}(\lambda)$. The starting point is the local diagonalization of the meromorphic differential system \cite{Wasowbook} (or equivalently the Birkhoff factorization), we will record the diagonalization for a general order $r_\infty$ and only then use it in our specific case: 
\begin{proposition}[Local diagonalization] Let $p$ be a point of the divisor $\mathcal{D}$ having a degree $r_p$ and $z_p(\lambda)$ a vanishing coordinate in the neighbourhood of the point $p$, then there exists a local gauge matrix $G_p$ such that the gauge transformation $\Psi_p = G_p  \hat{\Psi}$ provides 
\begin{align}
    \Psi_p (\lambda) = \Psi_p^{(\text{reg})} \diag\left(\exp\left(- \sum_{k=1}^{r_p-1}\frac{t_{p^{(1)},k}}{k z_p(\lambda)^{\,k}}  +t_{p^{(1)},0}\ln z_p(\lambda)\right) , \exp\left(- \sum_{k=1}^{r_p-1}\frac{t_{p^{(2)},k}}{k z_p(\lambda)^{\,k}}  +t_{p^{(2)},0}\ln z_p(\lambda)\right) \right) 
\end{align}
where $\Psi_p^{(\text{reg})}$ is regular in the neighbourhood of the pole $p$. This corresponds to a connection matrix given by 
\begin{align}
    G_p \,\hat{L}(\lambda) d\lambda\, G_p^{-1} + (\partial_\lambda G_p)G_p^{-1}d\lambda = dQ_p(z_p) + \Lambda_p \frac{dz_p}{z_p} 
 \qquad \text{where} \qquad Q_p(z_p) = \sum_{k=1}^{r_p-1} \frac{Q_{p,k}}{z_p^{\, k}}
\end{align}
The coefficients $Q_{p,k}$ contain the complex numbers which parameterize the irregular type of the system, while $\Lambda_p$ is referred to as the exponent of formal monodromy and contains the formal monodromies of the meromorphic differential system. They are given by 
\begin{align}
    Q_{p,k} =  \diag \left(- \frac{t_{p^{(1)},k}}{k},- \frac{t_{p^{(2)},k}}{k}\right) \qquad \text{ and } \qquad 
\Lambda_p = \diag(t_{p^{(1)},0},t_{p^{(2)},0})\,,\, \forall \, p\in \mathcal{D}\,.
\end{align}
One realizes effectively that for a tame pole with $r_p=1$ there is no irregular type.
\end{proposition}
We note that the set of formal monodromy is not the full monodromy since in our case, we have additional Stokes data that must be considered in sectorial analytic domains, we will not need this set for our purposes. We now proceed by adapting the above local diagonalization for our setting and using the conjugation action of the reductive group. In fact, in order to obtain an explicit expression, one needs to choose a representative of an orbit. Several choices exist and their advantages are dependent on the problem. We will fix here the natural choice of a diagonal leading order at infinity and we will not use the remaining action (of the stabilizer) to fix one of the entries of the leading order at the finite pole. In other words, our space is given by 
\begin{align}
      \widetilde{F}_{\mathcal{D},3} :=& \bigg\{ \widetilde{L}(\lambda) = \widetilde{L}_\infty + \frac{ \widetilde{ L}_{0}}{\lambda - X_0} \qquad
    \text{with} \qquad \left\{ \widetilde{L}_\infty, \widetilde{L}_0 \right\} \in    \mathfrak{gl}_2^2 \qquad  \text{and such that} \nonumber \\
   &   \widetilde{L}_\infty \, \,  \text{ and } \,\, \widetilde{L}_0 \,\,  \text{have distinct eigenvalues with the expression} \nonumber \\
    & \widetilde{L}_\infty = \diag (-t_{\infty^{(1)},1},-t_{\infty^{(2)},1}) \,\, \text{ and } \,\, \widetilde{L}_0 = \begin{pmatrix} \beta_{0} & \omega \\ \delta_0 & \gamma_0 \end{pmatrix}, \, (\beta_0,\delta_0,\gamma_0) \in \mathbb{C}^3 \bigg\} 
\end{align}
Furthermore, since $r_\infty=2$, we observe that the matrix $G_\infty(\lambda)$ is of the form $G_{\infty}(\lambda)=G_0+G_1\lambda^{-1}+O(\lambda^{-2})$ with $G_0$ diagonal (in order to preserve the fact that leading order is already diagonal). Consequently $G_{\infty}^{-1}\partial_\lambda G_\infty=O(\lambda^{-2})$. The diagonalization implies
\begin{align}
\label{SubLeadingCoeff2} \det (\widetilde{L}(\lambda)+G_\infty^{-1} \partial_\lambda G_\infty)\overset{\lambda\to \infty}{=}t_{\infty^{(1)},1}t_{\infty^{(2)},1}+(t_{\infty^{(1)},1}t_{\infty^{(2)},0}+t_{\infty^{(2)},1}t_{\infty^{(1)},0}) \lambda^{-1}+O(\lambda^{-2}).
\end{align}
Since $G_{\infty}^{-1}\partial_\lambda G_\infty=O(\lambda^{-2})$, the l.h.s. is of the form
\small{\begin{align} \label{Idrinftyequal2} \det (\widetilde{L}(\lambda)+G_\infty^{-1} \partial_\lambda G_\infty)&\overset{\lambda\to \infty}{=}&\begin{vmatrix} -t_{\infty^{(1)},1}+\beta_0\lambda^{-1}+O(\lambda^{-2})& \omega \lambda^{-1}+O(\lambda^{-2})\\
\delta_0 \lambda^{-1}+O(\lambda^{-2})& -t_{\infty^{(2)},1}+\gamma_0\lambda^{-1}+O(\lambda^{-2})\end{vmatrix} \nonumber   \\
&=&t_{\infty^{(1)},1}t_{\infty^{(2)},1}-(\gamma_0 t_{\infty^{(1)},1}+\beta_0 t_{\infty^{(2)},1}) \lambda^{-1}+O(\lambda^{-2}).
\end{align}}
\normalsize{Moreover}, the diagonalization of the singular part also implies that for $r_\infty=2$:
\begin{align} \label{Tracerinftyequal2}\tr  \widetilde{L}(\lambda)\overset{\lambda\to \infty}{=}-(t_{\infty^{(1)},1}+t_{\infty^{(2)},1})-(t_{\infty^{(1)},0}+t_{\infty^{(2)},0})\lambda^{-1}+  O(\lambda^{-2}).
\end{align}
Identifying \eqref{Idrinftyequal2} with the r.h.s. of \eqref{SubLeadingCoeff2} and using \eqref{Tracerinftyequal2} we get 
\begin{align} \beta_0=-t_{\infty^{(1)},0} \qquad \text{ and } \qquad \gamma_0=-t_{\infty^{(2)},0} \,.
\end{align}
We proceed by using the reduction of the above system following the action of the M\"obius transformations and the passage to traceless matrices $\mathfrak{gl}_2 \to \mathfrak{sl}_2$, removing the trace is naturally tensoring by a rank-one connection; it is not in general a globally single-valued rational scalar gauge. These actions have been carried out in details in \cite{MarchalOrantinAlameddine2022} covering the above case and we adapt this result to the present setting. 
\begin{lemma}[Cf. \cite{MarchalOrantinAlameddine2022}] The action of the M\"obius transformation and the passage to traceless matrices fixes, for all $p \in \mathcal{D}$, the following coefficients
\begin{align}
  t_{p^{(1)},k} =   -t_{p^{(2)},k} = t_{p,k}   \qquad \text{with} \qquad   X_0 = 0, \qquad \text{and} \qquad t_{\infty,1}=z
\end{align}
    where $z$ is taken to be a constant.
\end{lemma}
Note that when one passes to the traceless setup, the conjugation action of the reductive group $GL_2$ needs to be modified as well. In order to preserve the traceless property of the resulting gauge transformations, one has to satisfy 
\begin{align}
    \tr (\widetilde{A}) = - \tr (\partial_\lambda G G^{-1}) = - \partial_\lambda \log (\det G)
\end{align}
In other words, the determinant is taken to be a constant and thus passing from $GL_{2} (\mathbb{C})$ to $SL_{2} (\mathbb{C})$ up to a normalization of the determinant. Realize that in our system, we do not have any deformation parameters to consider isomonodromic deformations, indeed, this case corresponds to a $g= r-3 = r_\infty+1-3=0$, thus the symplectic fibration of the isomonodromic deformation parameters is of trivial dimension. Nevertheless, it is always possible to gauge transform the system (using automorphisms of the rank $2$ bundle) in order to pass to the oper gauge of the system and have explicit expressions, this is the main purpose of the following lemma. 
\begin{lemma} Let us define the following gauge matrix
\begin{align}
    G(\lambda):=\begin{pmatrix} 1&0\\
\widetilde{L}_{1,1}(\lambda)&\widetilde{L}_{1,2}(\lambda)
\end{pmatrix} = \begin{pmatrix}
    1 & 0 \\
    -t_{\infty,1}- \frac{t_{\infty,0} }{\lambda}& \frac{\omega}{\lambda}
\end{pmatrix}
\end{align}
then the gauge transformation $\Psi (\lambda):= G(\lambda ) \widetilde{\Psi}$ implies that $\Psi$ satisfies the companion like system
\begin{align}
    \label{GeneralLSys}\partial_\lambda \Psi(\lambda)=L(\lambda)\Psi(\lambda)\,, \qquad \text{with} \qquad L(\lambda)=\begin{pmatrix}0&1\\ L_{2,1}(\lambda)&L_{2,2}(\lambda) \end{pmatrix}
\end{align}
    with the non-trivial entries of the connection matrix given by 
    \begin{align} \label{entriesoper}
        L_{2,1}(\lambda) =  t_{\infty,1}^2  + \frac{H_{0,1}}{\lambda} + \frac{t_{0,0}^2}{\lambda^2}\,,          \qquad L_{2,2}(\lambda) = - \frac{1}{\lambda}
    \end{align}
\end{lemma}
where $H_{0,1}$ is a function of the times and monodromies (at infinity) of the system. Note in particular that the first line of the horizontal sections is unchanged by this gauge transformation and satisfies a quantum curve equation given by 
\begin{align}
    \left(\left[\partial_{\lambda}\right]^2 -L_{2,2}(\lambda)\partial_\lambda -L_{2,1}(\lambda)\right)\psi_i=0 \, ,\qquad i\in \{1,2\}
\end{align}
The above lemma is particularly useful since from it, one could recover for a general rank $2$ meromorphic system the explicit geometric connection matrix $\widetilde{L}(\lambda)$ via an analysis of the orders of the divisor, this has been done in details in \cite{Marchal_2024}. In our case, the missing piece is the entry $ \widetilde{L}_{2,1}$ of the original gauge which we record here before introducing the symmetry reduction. In fact, using the opposite gauge transformation to recover the original geometric matrix, one has 
\begin{align}
    \widetilde{L}_{2,1} (\lambda) = \frac{t_{0,0}^2 - t_{\infty,0}^2}{\omega \lambda} 
\end{align}

We are now set to perform the symmetry reduction on this system, in particular, we impose the following involution on the horizontal sections solution of the system 

\begin{definition} [Definition of the involution symmetry]\label{DefSymmetry}We impose the symmetry condition:
\begin{align}
\widetilde{\Psi}(-\lambda)= \sigma_1 \widetilde{\Psi}(\lambda)\sigma_1, \qquad \text{ where } \qquad  \sigma_1:=\begin{pmatrix} 0&1\\1&0\end{pmatrix} \
\end{align} 
on the horizontal sections of the $0-$element of the PIV system.
\end{definition}

We want to discuss, before proceeding to the impact of the symmetry on the Lax matrix, a general framework for equivariant connections whose horizontal sections admit a certain symmetry. The form given in \autoref{DefSymmetry} is actually an example of this general construction. For this we start by the following definition.

\begin{definition}[$\mathbb{Z}_2$-equivariant connection in a trivialization] Suppose that $E=X\times\mathbb{C}^2$ is a trivializable rank $2$ vector bundle over a Riemann surface $X$ and that the $\mathbb{Z}_2$-action is given by
\begin{align}
    \sigma:X\longrightarrow X, \qquad \sigma^2= I_X,
\end{align}
on the base and by
\begin{align}
    \widetilde{\sigma}(x,v)= \bigl(\sigma(x),S v\bigr), \qquad S\in GL_2(\mathbb{C}), \qquad S^2=I_2,
\end{align}
on the total space. Let the meromorphic connection be written as
\begin{align}
    \nabla=d-L, \qquad L \in\Omega^1(X,\mathfrak{gl}_2(\mathbb{C}))\,\, \text{is a meromorphic $1-$form}.
\end{align}
Then $\nabla$ is said to be $\mathbb{Z}_2$-equivariant if and only if
\begin{align} L = S\,\sigma^*L \,S^{-1}\, \qquad \text{equivalently,}  \qquad\sigma^*L=S^{-1}LS.
\end{align}
\end{definition}
Note that in the above definition, we are assuming that the action on the base $X$ admits a lift with respect to the canonical projection, and that the points of the divisor of the connection are left invariant under this action. This definition gives for a coordinate $x$  
\begin{align}
    L(\sigma(x))\,d\sigma(x) = S\,L(x)\,S^{-1}\,dx.
\end{align}
This recovers exactly the involution that we consider above, we have however a more general definition that is easily generalizable to any rank $d$ that we hope to come back to in the future for generalizing group actions and equivariant connections. In rank $2$ the symmetry condition of \autoref{DefSymmetry} actually breaks down to two cases that are given as follows. 
\begin{lemma}
Let $S\in GL_2(\mathbb C)$ satisfy $
S^2=I_2$, 
and suppose that a fundamental solution $\Psi(\lambda)$ satisfies the 
symmetry condition
\begin{align}
    \Psi(-\lambda)=S\Psi(\lambda)S^{-1}.
\end{align}
Then, up to a constant gauge transformation and multiplication of $S$ by a nonzero scalar, there are only two possible symmetry types:
\begin{enumerate}
    \item The trivial case
    \[
    S=I_2,
    \qquad
    \Psi(-\lambda)=\Psi(\lambda),
    \qquad
    L(-\lambda)=-L(\lambda);
    \]
    \item The non-trivial case (the case considered here)
    \[
    S=\sigma_3=
    \begin{pmatrix}
    -1&0\\
    0&1
    \end{pmatrix},
    \qquad
    \Psi(-\lambda)
    =
    \sigma_3\Psi(\lambda)\sigma_3^{-1},
    \]
    and hence
    \[
    L(-\lambda)
    =
    -\sigma_3L(\lambda)\sigma_3^{-1}.
    \]
\end{enumerate}
\end{lemma}
\begin{proof}
Since $S^2=I_2$ the minimal polynomial of $S$ divides $(x-1)(x+1)$ which has distinct roots over $\mathbb C$. Hence $S$ is diagonalizable
with eigenvalues in $\{\pm1\}$. Moreover, the symmetry condition is unchanged under
\begin{align}
    S\mapsto cS,\qquad c\in\mathbb C^*
\end{align}
since the scalar factor cancels between $S$ and $S^{-1}$. In particular, we may replace $S$ by $-S$ and assume that one of its eigenvalues is $+1$. Thus, after a constant change of basis, there are only two possibilities: $S=I_2$ or $S=\sigma_3.$ Indeed, if $D$ is a matrix of eigenvectors such that $D^{-1}SD=V$, then the constant gauge transformation
\begin{align}
\widehat{\Psi}=D^{-1}\Psi D, \qquad \text{provides} \qquad
\widehat{\Psi}(-\lambda)
=
V\widehat{\Psi}(\lambda)V^{-1}.
\end{align}
The corresponding transformations of $L$ follow by differentiating
with respect to $\lambda$.
\end{proof}

More generally, an equivariant connection has horizontal sections satisfying $\Psi (-\lambda) = S \Psi(\lambda) C$ with a suitable choice for the constant matrix $C$, in fact the choice taken here is not arbitrary. The presence of $\sigma_1$ on the right is crucial for the local Birkhoff factorizations to hold even if it does not modify the symmetry condition for the connection matrix which is really the one that we will use in this article. Note that the symmetry is an involution and depends on the gauge in the sense that it is imposed to $\widetilde{L}$ (it would take other forms in different gauges). Note that, since the symmetry governs the transformation $\lambda \mapsto -\lambda$, it is not compatible with arbitrary M\"obius transformations acting on $\lambda$. For example, for higher elements of the PIV hierarchy, if the position of the finite pole had been taken arbitrary (and fixing $t_{\infty,2d}=0$ as in the canonical choice made in \cite{Marchal_2024}) then the symmetry condition would have needed an adaptation. As we will see below, this symmetry is compatible with the normalization at infinity of $\widetilde{L}(\lambda)$, i.e. the fact that $\widetilde{L}(\lambda)=\text{diag}(-t_{\infty,1},t_{\infty,1}) +O(\lambda^{-1})$. However, it will fix the value of $\omega$ to a non-trivial value thereby leaving no remaining degrees of freedom for the $SL_2(\mathbb{C})$ action (instead of $GL_2(\mathbb{C})$ since we have used the passage to traceless matrices). The symmetry considered in this section admits a natural geometric interpretation. At the level of the spectral parameter, the involution $\lambda \mapsto -\lambda$ preserves the pole structure of the entire PIV hierarchy and acts on the corresponding space of meromorphic connections. Requiring the horizontal sections to satisfy $\widetilde{\Psi}(-\lambda)=\sigma_1 \widetilde{\Psi}(\lambda)\sigma_1$ and select the fixed-point locus of this action. As a consequence, the associated spectral data acquire an additional symmetry and naturally descend to the quotient by $\lambda \mapsto -\lambda$. From this perspective, the Flaschka-Newell Painlev\'e II hierarchy appears as the isomonodromic dynamics induced on the invariant submanifold determined by the involution. We will now show how the $0-$element recovers the Bessel system in a similar way after using the reduction. In particular, the main result of this section is the following 

\begin{proposition}[The Bessel system] \label{BesselsysProp} The symmetry of \autoref{DefSymmetry} translates to the Lax matrix as $\widetilde{L}(-\lambda) = - \sigma_1 \widetilde{L}(\lambda) \sigma_1$ and imposes the following constraints
\begin{align}
    t_{\infty,0} = 0, \qquad \omega= t_{0,0}\,.
\end{align}
This transforms the system into the Bessel system, i.e. the horizontal sections can be expressed explicitly using Bessel functions.
\end{proposition}
\begin{proof}
    The symmetry condition implies that the entry $[1,1]$ of the Lax matrix is an even function, the direct computation forces $t_{\infty,0} = 0$. The remaining conditions are forced by the swapping with the minus signs of the off-diagonal entries and a choice has been taken for $\omega^2 = t_{0,0}^2$. This leads to the following matrix in the original gauge
    \begin{align}
        \widetilde{L}^{\text{red}}(\lambda) =& \begin{pmatrix}
            -t_{\infty,1} & \frac{t_{0,0}}{\lambda} \\
            \frac{t_{0,0}}{\lambda} & t_{\infty,1}
        \end{pmatrix} \qquad \text{which after a rescaling of the form} \quad \psi_1 \pm \psi_2 = \sqrt{\lambda} \phi_\pm  \nonumber  \\
      &  \left(\left[\partial_{x}\right]^2 +  \frac{1}{x}   \partial_x  + \left(- 1 - \frac{\left( t_{0,0} \mp\frac{1}{2}  \right)^2 }{x^2} \right)\right)\phi_\pm=0 \, ,
    \end{align}
    where the variable $x= t_{\infty,1} \lambda$ was introduced. This is the modified Bessel equation satisfied by the horizontal sections after setting $c=t_{\infty,1}$ (a constant parameter) and the monodromy controlling the solution in the neighborhood of the finite pole $t_{0,0}=\alpha$. By pulling back this oper system to the original gauge (and rescaling), one leaves the first line of the horizontal sections invariant and thus preserves the two Bessel functions solutions of the Bessel differential equation. 
\end{proof}
\begin{remark}
    The above system is also naturally related to the Whittaker equation, the system is constant-gauge equivalent to a modified Bessel system through the change of variables and rescaling introduced in the proof. Its scalar reductions in the original basis are Whittaker equations with parameters \(\kappa=\pm\tfrac12\) and \(\mu=\alpha\).
\end{remark}

The Bessel equation is known to be related to a “folded" differential equation, namely the Bessel-Clifford equation as discussed in the introduction. We will make this more precise in geometric terms in the following subsection when we define the folding operation.

\subsection{Folding and the Bessel-Clifford system} \label{sec2.2}
We will start by recalling the classical example of the folding that is mostly understood in the literature, for this, we consider the Flaschka-Newell Lax representation for the Painlev\'e II equation, which is the first element of the hierarchy. The system is given by 
\begin{align}
    \frac{\partial \Psi_{FN}}{\partial \lambda} =& L_{FN} \Psi, \qquad \text{and} \qquad \frac{\partial \Psi_{FN}}{\partial x} = A_{FN} \Psi \nonumber\\
    L_{FN} (\lambda,x)  =& -i \left( 4 \lambda^2 + x + 2 u^2 \right) \sigma_3 + \left( 4 \lambda u + \frac{\alpha}{\lambda} \right) \sigma_1 - 2 v \sigma_2 \nonumber \\
    A_{FN} (\lambda,x) = & - i \lambda \sigma_3 +  u \sigma_1
 \end{align}
 here the matrices $\sigma_i$ are the usual Pauli matrices. The compatibility of the above system (commutation of the flows) provides the Painlev\'e II equation as its condition. Moreover, it is clear that in the generic setup, this corresponds to an irregular pole at $\{\infty\}$ of order $r_\infty=4$ and a tame pole at the origin of order $r_0=1$. However, this exact case in the generic setup provides the second element of the PIV hierarchy. This observation has been made explicit in \cite{alameddine2026symmetryreductionpainleveiv} and provided a reduction of the system to cover the Painlev\'e II hierarchy as an invariant locus. In particular, the symmetry condition is exactly the symmetry imposed in \autoref{DefSymmetry} and has an invariant locus providing the Flaschka-Newell Painlev\'e II representation (and this translates to the whole hierarchy). This symmetry condition with respect to $\lambda \mapsto - \lambda$ implies the existence of a Lax pair on the “folded" complex $\lambda-$plane. Indeed, introducing the folded coordinate and the gauge transformation 
 \begin{align}
     G(\lambda) := \begin{pmatrix}
         \lambda^{\frac{1}{2}} & 0 \\
         0 & \lambda^{-\frac{1}{2}}
     \end{pmatrix} S      \qquad \text{with} \qquad  \xi := \lambda^2
 \end{align}
where the matrix 
\begin{align}
    S: = \frac{1}{\sqrt{2}} \begin{pmatrix}
        1 & -1 \\ 1 & 1
    \end{pmatrix}
\end{align}
is the matrix diagonalizing the symmetry condition from $\sigma_1 \to \sigma_3$. Our transformed system now has the following Lax pair
 \begin{align} \label{newsys}
     \tilde{L} = \begin{pmatrix}
         2 u + \frac{\alpha + \frac{1}{2}}{2 \xi} & 2 i \xi + i q \\2 i + \frac{ir}{\xi } & - \left( 2 u + \frac{\alpha + \frac{1}{2}}{2 \xi} \right) \end{pmatrix}, \qquad \tilde{A} = \begin{pmatrix}
             u & i \xi \\ i & -u
         \end{pmatrix}
 \end{align}
 The additional parameters are given by $ q=v+ u ^2 + \frac{x}{2}$ and $r= - v + u^2+ \frac{x}{2}$. Now evaluating the compatibility of the new system provides the following known result. 
\begin{lemma} The compatibility condition evaluated on the system \eqref{newsys} provides the following set of equations 
\begin{align}
    \partial_x u = \frac{q-r}{2}, \qquad \partial_x q = 2 u q - \alpha + \frac{1}{2}\qquad \partial_x r = - 2 u r + \alpha + \frac{1}{2}
\end{align}
    in particular, the system implies $q+r - 2 u^2 -x$ is actually vanishing and that eliminating $q$ and $r$ provides the equation $\partial_x^2 u = 2 u^3 + xu - \alpha$ which is the second Painlev\'e equation. Moreover, eliminating $u$ and $q$ (up to some rescaling) provides the equation number XXXIV in the list of Painlev\'e and Gambier: 
    \begin{align}
        \partial_x^2 r  = \frac{\left( \partial_x r \right)^2}{2r} + 2 r^2 - xr - \frac{\left(   \alpha+ \frac{1}{2}  \right)^2}{2r}, \qquad u = - \frac{\partial_x r }{2r} + \frac{\alpha+ \frac{1}{2}}{2 r}
    \end{align}
    or a similar equation for $q$ and $u$ if one eliminates $u$ and $r$. 
\end{lemma}
In fact, the system \eqref{newsys} is used as the natural non-symmetric version of the Flaschka-Newell Lax pair for Painlev\'e II, using the symmetric version or the folded one remains dependent on the nature of the problem and the structure of the Lax system needed. Before proceeding, we make the following definition.

\begin{definition}[Spectral curve] The spectral curve attached to a (traceless) connection matrix $L(\lambda)$ is the affine plane curve defined by the characteristic equation 
\begin{align}
    \Gamma(y, \lambda) := \left\{ \det\left( y - L(\lambda) \right)=0 \right\} = \left\{ y^2 = - \det\left( L(\lambda) \right) \right\}
\end{align}
This characteristic equation defines the eigenvalues as functions on the corresponding $2$-sheeted Riemann surface.     
\end{definition}

\begin{remark}
    Note that saying that the Flaschka-Newell representation is the genus $1$ element of the hierarchy is actually not accurate, since the spectral curve in this setting is given by 
    \begin{align} \label{FlaschkaNewellSpectralCurve}
        y^2 = - 16 \lambda^4 - 8 x \lambda^2 - \left( x + 2u^2 \right)^2 + 4 v^2 +  8 \alpha u + \frac{\alpha^2}{\lambda^2}
    \end{align}
    it is in particular of genus $2$ if we eliminate degenerations despite the fact that its Lax system corresponds to a representation of a Painlev\'e equation (usually these are known as the genus $1$ examples having a symplectic leaf of dimension $2$). The main difference in these cases is that the symmetry of the Lax system under the involution imposes the reduction of the dimension of the symplectic leaf thus recovering the natural expected dimension (and genus). More generally, this is the case of all the members of the Flaschka-Newell hierarchy where the spectral curve is given by 
    \begin{align}
        y^2  = -4^{2g} \lambda^{4g} + W_{2g-1} (\lambda^2) + \frac{\alpha_g^2}{\lambda^2} \qquad \qquad
        \text{with a degree $2g-1$ polynomial $W$,}
    \end{align}
    which is a genus $2g$ spectral curve representing the $g$ member of the hierarchy.  
\end{remark}

Our strategy for the folding is similar to the above construction of the folded Lax pair, in our case, we will recover the Bessel-Clifford system which provides the realization of the folding as a relation between the Bessel functions and the Bessel-Clifford functions (or equivalently Kummer's ${}_0F_1$ function). This is the main result of the section.

\begin{theorem} \label{TheoFolding}
    Consider the Bessel system defined and normalized on the $\lambda-$sphere given by
    \begin{align} \label{eq2-30}
        \partial_\lambda \widetilde{\Psi} = \widetilde{L}(\lambda)\widetilde{\Psi} \qquad \text{with} \qquad \widetilde{L}(\lambda) = \begin{pmatrix}
            - z & \frac{\alpha}{\lambda} \\ \frac{\alpha}{\lambda} & z
        \end{pmatrix}
    \end{align}
   It admits a connection having a $\mathbb{Z}_2-$equivariant structure covering the involution $(\lambda, \widetilde{\Psi}) \to (- \lambda, \sigma_1 \widetilde{\Psi} \sigma_1)$, and after diagonalizing the involution, descends into a rational connection on the folded $\xi-$sphere where $\xi= \lambda^2$. Moreover, the quotient system admits solutions that are scalar-equivalent to the solutions of the Bessel-Clifford, or confluent hypergeometric limit, or kummer's second, equation 
    \begin{align}
         u f''(u) +  \beta f'(u) - f(u) =0
    \end{align}
    Thus, the Bessel-Clifford connection arises as a scaled reduction of the $\mathbb{Z}_2-$folded connection associated with the original symmetric Bessel connection.
\end{theorem}

\begin{proof}
    First, it is clear that the connection matrix satisfies $\widetilde{L}(-\lambda) = - \sigma_1 \widetilde{L}(\lambda) \sigma_1$, hence the connection $1-$form itself is invariant under the combined action $(\lambda, \Psi) \mapsto (-\lambda, \sigma_1 \Psi \sigma_1)$. Therefore, the connection is $\mathbb{Z}_2-$equivariant, this establishes the existence of the folding of the $\lambda-$plane. To establish the folding, one needs to diagonalize the symmetry by using the fact that $S \sigma_1 S^{-1} = \sigma_3$, indeed, one has
    \begin{align} 
        \check{L}(\lambda) = S\widetilde{L}(\lambda)S^{-1}, \qquad \text{which provides} \qquad \check{L}(\lambda) = \begin{pmatrix} \frac{\alpha}{\lambda} & -z \\ - z & - \frac{\alpha}{\lambda}
        \end{pmatrix} 
    \end{align}
    The main feature of the above gauge is that the diagonal entries are now odd functions and the off-diagonal entries are even with the symmetry that does not swap entries. Introducing now the gauge transformation along with the folded coordinate 
    \begin{align}
        G(\lambda) :=  \begin{pmatrix}
            \lambda^{\frac{1}{2}} & 0 \\ 0& \lambda^{-\frac{1}{2}}
        \end{pmatrix}, \qquad \text{with} \qquad \xi= \lambda^2
    \end{align}
    where the square roots are understood locally on the punctured $\lambda-$plane, their role is to produce a rational connection after passing to the quotient. In fact, the system after this becomes 
    \begin{align}
     \partial_\xi Y(\xi) =  L^{\text{fold}}(\xi) Y(\xi)   \qquad \text{with} \qquad L^{\text{fold}}(\xi) = \frac{1}{2} \begin{pmatrix}
            \frac{\alpha-\frac{1}{2}}{\xi} & -\frac{z}{\xi} \\ - z & -   \frac{\alpha-\frac{1}{2}}{\xi} 
        \end{pmatrix}
    \end{align}
    where the connection $1-$form transforms under the change of coordinates as $ 2 \lambda\check{L}(\lambda) d \lambda = L^{\text{fold}}(\xi) d \xi$. We will now turn our attention to prove that the above system admits solutions that are given by kummer's ${}_0F_1$ function, which effectively means that our system is the Bessel-Clifford system. There are at least two equivalent ways to proceed, either using the oper gauge (and thus the quantum curve) in analogy with the previous Bessel system, or by scalar reduction of the above system. We will use the latter method, in fact one has by using the two equations of the system 
    \begin{align} \label{Bessel-Cliffordcurvequantum}
        \begin{cases}
            y'_1 = \frac{\alpha- \frac{1}{2}}{2 \xi} y_1 - \frac{z}{2 \xi} y_2 \\
            y_2' =  - \frac{z}{2 } y_1- \frac{\alpha- \frac{1}{2}}{2 \xi} y_2
        \end{cases} \qquad \text{providing} \qquad \xi y''_1 + y'_1 - \left( \frac{z^2}{4} + \frac{\left( \alpha-\frac{1}{2} \right)^2}{4 \xi} \right) y_1=0 
    \end{align}
    Then, we introduce the substitution $y_1(\xi) = \xi^{\frac{\alpha}{2} - \frac{1}{4}} f(\xi)$ and putting this in the equation and dividing by $\xi^{\frac{\alpha}{2}- \frac{1}{4}}$, one obtains the equation 
    \begin{align}
        \xi f'' + \left( \alpha+ \frac{1}{2} \right) f' - \frac{z^2}{4} f=0
    \end{align}
    The final step is the introduction of the rescaling variable $u=\frac{z^2}{4} \xi$, which transforms the above equation into
    \begin{align}
        u f''(u) +  \left(\alpha + \frac{1}{2} \right) f'(u) - f(u) =0
    \end{align}
    which is nothing but the Bessel-Clifford equation, in other words, a solution is given by Kummer's function
    \begin{align} 
        f(u) = {}_0 F_1 \left( ;\alpha + \frac{1}{2} ; u \right) \qquad \text{providing} \qquad y_1 (\xi) = \xi^{\frac{\alpha}{2} - \frac{1}{4}} {}_0 F_1 \left( ;\alpha + \frac{1}{2} ; \frac{z^2 \xi}{4} \right)
    \end{align}
Thus completing the proof. 
\end{proof}
\begin{remark}
    Realize that after the folding, the leading order of the connection matrix is non-diagonalizable, in analogy with the Flaschka-Newell case where one descends into a connection whose leading order admits a triangular form. This means that the connection in this case is ramified (seen directly from the exponential behavior as well), indeed, it is very well known that the solutions to the Bessel-Clifford equation admit a ramification and a fractional exponent $x^{\frac{1}{2}}$, this justifies the interpretation of the symmetry reduction from the even Painlev\'e IV hierarchy to the Flaschka-Newell hierarchy as a twisting procedure after the introduction of a folding operation on the invariant locus. This could be equivalently seen from the irregular type of the Bessel system when introducing the folding variable.
\end{remark}

\begin{remark}
    One has the residue at the tame pole admitting two eigenvalues given by $\rho_1 = \frac{\alpha}{2} - \frac{1}{4}$ and $\rho_2 = - \frac{\alpha}{2} + \frac{1}{4}$, the first exponent is exactly the exponent of the scalar solution $\xi^{\frac{\alpha}{2}- \frac{1}{4}}$, this provides a connection-related explanation for the factor in front of the solution. The result is a geometric recovery of the well known passage from solutions of the Bessel-Clifford equation to the solutions of the Bessel equation through a double cover. We note here that the gauge transformation cannot cause the additional half integer twists in the monodromy of the source, in fact, as discussed in the previous subsection when establishing the relation to the Bessel equations, one needs to introduce a half-integer twist in $\lambda$, this is the main factor causing the shift of the monodromy.
\end{remark}
The proof of \autoref{TheoFolding} contains a certain rescaling of the solutions, in fact, by introducing the rescaling matrix 
\begin{align}
    J (\xi)= \begin{pmatrix}
        \xi^{\frac{\alpha}{2} - \frac{1}{4}} & 0 \\ 0 & \xi^{-\left(\frac{\alpha}{2} - \frac{1}{4} \right) }
    \end{pmatrix} 
\end{align}
One obtains a scaled Lax matrix giving the spectral curve
\begin{align}
    \widetilde{L}^{\text{fold}}(\xi) = \frac{1}{2} \begin{pmatrix} 0 &  - z \xi^{ -\alpha - \frac{1}{2}}  \\ - z \xi^{ \alpha-\frac{1}{2}} & 0 
        \end{pmatrix} \qquad \text{providing the spectral curve} \qquad \xi y^2 = \frac{z^2}{4}
\end{align}
which is referred to in the TR community as the Bessel spectral curve \cite{DoNorbury2018} for a certain constant value of $z$ (nevertheless, it is in direct relation with the Bessel-Clifford system). This spectral curve for another value of the trivial times $z=2$ was also studied in \cite{DoNorbury18} where an irregular version of TR has been introduced and provided polynomial volumes on a certain moduli space in analogy with the volumes considered by M. Kontsevich on the moduli space of curves \cite{Kontsevich1992}. We stress here that the different gauge is not to be thought of as an alternative system that could be used for TR without a price to pay. In fact, a similar gauge transformation for the matrix $\check{L}(\lambda)$ provides a trivial spectral curve $\mu^2 =z^2$. This is due to the fact that the spectral data are not invariant under gauge transformations. If one needs to define a folding operation for TR, one needs to consider the original geometric gauge, and not a different (perhaps more suitable or trivial) gauge.  

\medskip 

Let us now justify our claims in the introduction by considering the passage to the Schr\"odinger gauge for meromorphic differential systems. This ultimately proves the passage from the Kummer to Bessel curves.  

\subsection{The Schr\"odinger gauge} \label{SubsecSchrod}

We will first show how one can pass from the oper gauge to the Schr\"odinger gauge for a general rank $2$ meromorphic differential system. We will then restrict the discussion to the cases considered in the previous two sections. Suppose we have a rank $2$ system in its oper form satisfied by the horizontal sections, this means that the horizontal sections are solutions of the quantum curve
\begin{align}
    \left(\left[\partial_{\lambda}\right]^2 -L_{2,2}(\lambda)\partial_\lambda -L_{2,1}(\lambda)\right)\psi_i=0 \, ,\qquad i\in \{1,2\}
\end{align}
The coefficients of this equation are expressed in function of the original (geometric gauge) via the following expressions 
\begin{align}
     L_{2,1}=&-\det \widetilde{L}+\partial_\lambda\widetilde{L}_{1,1}-\widetilde{L}_{1,1}\frac{\partial_\lambda\widetilde{L}_{1,2}}{\widetilde{L}_{1,2}}, \qquad 
L_{2,2}=  \tr \widetilde{L} +\frac{\partial_\lambda\widetilde{L}_{1,2}}{\widetilde{L}_{1,2}}.
\end{align}
For general higher genus curves $g \geq 1$, this is the gauge that allows the definition of the apparent singularities as the zeroes of the entry $\widetilde{L}_{1,2}$. By applying an alternative gauge transformation
\begin{align}
    \phi_1= G\psi_1, \qquad \text{with} \qquad G:= \exp\left( - \frac{1}{2} \int L_{2,2} (\lambda) d \lambda   \right)
\end{align}
the quantum curve equation transforms into a Schr\"odinger type equation satisfied by the horizontal sections 
\begin{align}
    \left( \left[\partial_{\lambda}\right]^2 - Q(\lambda) \right) \phi_i = 0 \qquad \text{with} \qquad Q(\lambda) := L_{2,1}(\lambda) + \frac{L_{2,2}(\lambda)^2}{4} - \frac{\partial_\lambda L_{2,2}(\lambda) }{2}
\end{align}
\begin{remark}
    In the case of a genus $g\geq 1$ one has an additional deformation equation that changes in the Sch\"odinger gauge, in particular, if one has for an isomonodromic deformation direction on the tangent space of times the equation 
    \begin{align}
         \frac{ \partial \Psi}{\partial t_i} = A_i(\lambda,\mathbf{t}) \Psi
    \end{align}
    We will take this gauge to be the oper gauge, then the gauge transformation providing the Schr\"odinger gauge transforms the above system to the equation 
    \begin{align}
    \frac{ \partial \phi_1}{\partial t_i} = [ A_i ]_{1,2}   \frac{ \partial \phi_1}{\partial \lambda} + \left( [ A_i ]_{1,1} + \frac{ [ A_i ]_{1,2} L_{2,2} }{2} + \frac{1}{G} \frac{\partial G}{\partial t_i} \right) \phi_1
    \end{align}
    By taking an appropriate normalization for the gauge matrix $G$ we can transform the above equation into the form 
    \begin{align}
        \frac{ \partial \phi_1}{\partial t_i} =  \left( \frac{ [ A_i ]_{1,2} }{ \widetilde{L}_{1,2}}  \frac{\partial}{\partial \lambda}  -\frac{1}{2}  \left( \frac{\partial}{\partial \lambda} \frac{ [ A_i ]_{1,2} }{ \widetilde{L}_{1,2}}  \right) \right) \phi_1. 
    \end{align}
    This is the analogous Schr\"odinger gauge deformation equation.
\end{remark}
This offers a direct way to obtain the Schr\"odinger gauges of both Bessel and Bessel-Clifford systems by applying the above gauge transformation, we get the following result in relation with the results of \cite{Iwaki2019}. 

\begin{lemma}
    The combination of the symmetry and the folding operations realized from the Bessel to the Bessel-Clifford systems reproduces the same passage of the confluence diagram from Kummer to Bessel of spectral curves obtained as classical limits of the confluent family of the Gauss hypergeometric equations. This recovery of the spectral curves is up to an identification of the monodromies and parameters of the potentials of the respective curves of Kummer and Bessel (as in \cite{Iwaki2019}). 
    \end{lemma}
\begin{proof}
    By applying the formulas of the Schr\"odinegr gauge directly to the expressions in \eqref{entriesoper}, one obtains the potential 
    \begin{align}
        Q(\lambda)= \frac{ 4t_{\infty,1}^2 \lambda^2 + 4 H_{0,1} \lambda + 4 t_{0,0}^2 -1  }{4 \lambda^2} 
    \end{align}
    which identifies with the Kummer potential given in \cite{Iwaki2019} up to $\lambda=x$, and the monodromy parameters $ H_{0,1} = \lambda_\infty$ and $4t_{0,0}^2 - 1=4\lambda_0^2$, note that we have fixed the term $t_{\infty,1}^2$ to a constant. Applying the same procedure to go to the oper gauge and then to the Schr\"odinger gauge provides in the $\xi$ variable the potential 
    \begin{align}
        Q(\xi) = \frac{z^2 \xi + \left( \alpha + \frac{1}{2} \right)^2 -1}{4 \xi^2}
    \end{align}
    Identifying with the potential given in \cite{Iwaki2019} for the Bessel case gives $\xi=x$ and the monodromy $\left( \alpha + \frac{1}{2} \right)^2 -1 = 4\lambda_0^2$. This completes the proof. 
\end{proof}
This completes the analysis of the spectral curves and the different gauges. We will now turn our attention to the TR analysis of the spectral curves of both Lax matrices associated to the Bessel and the Bessel-Clifford systems. We will use, for simplicity, the curves given in \cite{Iwaki2019} since they give a rather simpler version. These curves are the same as the ones considered in the Lax systems due to the above identification.

\section{Equivariant topological recursion} \label{sec3}

\subsection{Curves with a symmetry and equivariant topological recursion} \label{sec3.1}
We first start by recalling the main definition of TR as introduced by Eynard-Orantin \cite{Eynard:2007kz}. TR is a universal procedure that associates to a spectral curve a collection of symmetric multidifferentials, a curve with some extra data. It has a wide variety of applications ranging from enumerative geometry, matrix models and Hurwitz theory to integrable systems and Lax representations of integrable hierarchies. First of all, as any recursion, TR needs a set of initial data 
\begin{definition}\label{DefTR}
    The initial data of the TR is a spectral curve $\mathcal{S} = \left( \Sigma,x,y,B \right)$ consisting of
    \begin{itemize}
        \item A Riemann surface $\Sigma$, not necessarily connected or compact;
        \item  a map $x: \Sigma \to \mathbb{C}$ such that its differential $dx$ is meromorphic and has finitely many simple zeroes denoted $a_1, \dots a_N$;
        \item a meromorphic form $dy $ on $\Sigma$  that is holomorphic and non-vanishing  at the zeroes of $dx$ or has at most a double pole;
        \item a symmetric bi-differential $B^{TR}$ on $\Sigma \times \Sigma$ with a polar structure including only a double pole on the diagonal with bi-residue $1$. 
    \end{itemize}
    Then TR produces symmetric multidifferentials on $\Sigma^n$ for $2g-2+n>0$ via the recursion 
    \begin{align}
        \omega_{g,n}(z_1, \dots, z_n) =  & \sum_{i=1}^N  \underset{ z = a_i }{\Res} \,  K_i(z_I,z)     \bigg(  \omega_{g-1,n+1} (z, \sigma_i(z) , z_2, \dots,z_n ) \nonumber  \\\label{EOTRtdom}
     &  + \sum_{\substack{g_1 + g_2 = g, \,\,\,  I_1 \bigsqcup I_2 = \{2,\dots, n \}  \\ (g_i, \vert I_i \vert  ) \neq (0,0) }}  \omega_{g_1, \vert I_1 \vert +1  } (z, z_{I_1}) \omega_{g_2, \vert I_2 \vert +1 } (\sigma_i(z), z_{I_2})      \bigg)
    \end{align}
    where $K_i$ is the recursion kernel defined in the neighbourhood of the point $a_i$ via
    \begin{align}
        K_i(z_1,z) := \frac{\frac{1}{2} \int_{w=\sigma_i(z)}^z   B^{TR}(z_1,w)    }{\left( y(z ) - y(\sigma_i(z) ) dx(z) \right)}
    \end{align}
    In the definition $\sigma_i: U_i \to U_i$ is the Galois involution near the ramification point $a_i \in U_i$. The first two differentials are given by $\omega_{0,1} = ydx$ and $\omega_{0,2} = B^{TR}$.
\end{definition}
It is well known that the differentials $\omega_{g,n}$ produced by the recursion are symmetric meromorphic multidifferential on $\Sigma^n$ with poles only on the ramification points. The bidifferential $B^{TR}$ can actually have an arbitrary non-vanishing bi-residue at the double pole on the diagonal (not necessarily 1). A rescaling of the form $B^{TR} \to c\, B^{TR}$ corresponds to a rescaling of the correlators $\omega_{g,n} \to c^{3g-3+2n} \omega_{g,n}$. 

eq-TR is a version of TR adapted to spectral curves equipped with the action of a finite group \cite{Giacchetto:2021atk}. The group
action induces additional symmetries of the spectral data, which are incorporated into the recursive construction. In the present setting,
the relevant group is $G=\mathbb{Z}_2$.

\begin{definition}[Equivariant topological recursion]\label{def:equTR} Let $G$ be a finite group, the initial data for eq-TR consists of an equivariant spectral curve which is a tuple $\mathcal{S}^{eq} = \{ \Sigma, \varphi, x, y, B, \chi, \upsilon, \beta  \}$ consisting of 
\begin{itemize}
    \item a Riemann surface $\Sigma$ with a group action $\varphi : G \times\Sigma \to \Sigma$ which is denoted by $\varphi_\gamma x = \gamma x$ for an element $\gamma 
    \in G$.
    \item a map $x: \Sigma \to \mathbb{C}$ such that its differential is meromorphic and has finitely many simple zeroes denoted $a_1, \dots a_N$;
    \item a meromorphic form $dy: \Sigma \to \mathbb{C}$ such that it does not vanish or has at most a double pole at the zeroes of $dx$;
    \item a symmetric bidifferential $B$ on $\Sigma \times \Sigma$;
    \item the group action $\varphi$ is free in the neighbourhood of the points $a_1, \dots a_N$; 
    \item three one-dimensional representations $\chi \,, \, \upsilon\,, \beta : G \to \mathbb{C}^\times $ which we refer to as characters and are determined, for any element $\gamma \in G$, by
    \begin{align}
        d x (\gamma z) = \chi_\gamma dx(z), \qquad y(\gamma z) = \upsilon_\gamma y(z), \qquad B(\gamma z_1,z_2) = \beta_\gamma B(z_1,z_2)
    \end{align}
    where $B(z_1,z_2) - B_G^\beta (z_1,z_2)$ is regular as $z_1 \to \gamma z_2$ with
    \begin{align}\label{bextendedregularity}
        B_G^\beta (z_1,z_2) :=   \sum_{\eta \in G} \beta_\eta^{-1} B^{TR}(\eta z_1,z_2)
    \end{align}
    \end{itemize}
    The recursion determining the higher multidifferentials is defined as for the ordinary TR of \autoref{DefTR}, but replacing $\omega_{0,1}=ydx$ by $\omega_{0,1}=\frac{1}{|G|}ydx$, and $B^{TR}$ by $B$. In particular, this affects $\omega_{0,2}$ and, thus, the recursion kernel $K_i$.

\end{definition}
\begin{remark}\label{rmk:eqTRscaling}
Our definition of $B$ in eq-TR differs from the one in \cite{Giacchetto:2021atk} by a factor of $|G|$, that is, $B\to |G|B$, see \eqref{bextendedregularity}. This rescales the generated $\omega_{g,n}$ by a factor of $|G|^{3g+2n-3}$ (not including $\omega_{0,1}$), which, however, after a further rescaling of $\omega_{0,1}=ydx\to |G|ydx$, provides an overall rescaling $\omega_{g,n}\to |G|^{g+n-1}\omega_{g,n}$; see also \cite[Lemma 9.1]{Giacchetto:2021atk} for details. For enumerative interpretations, this rescaling only gives an overall factor for fixed $g,n$. However, we will be interested in deriving a quantum curve, where, in the constructed wave function, multidifferentials $\omega_{g,n}$ with fixed $k=2g+n-2$ are collected together. Here, a scaling by $|G|^{g+n-1}$ is very unnatural and does not give a global transformation of either the wave function or the quantum curve. As we will see later, our \autoref{def:equTR} (which was already used in \cite{Alexandrov:2025kzn}) is the most natural one for deriving the quantum curves considered here. None of the properties proved in \cite{Giacchetto:2021atk} are affected by this rescaling.
\end{remark}

From the above definition, there are some consequences that are forced on the characters (see \cite{Giacchetto:2021atk} for details), for instance, the character $\beta$ must satisfy $\beta^2=1$, equivalently $\beta:G \to \{ \pm 1\}$. Then, the multidifferentials generated by eq-TR can just be non-vanishing only if the characters satisfy $\chi . \upsilon = \beta$. Thus, we will restrict ourselves to this case. Assuming this, the multidifferentials satisfy the group action
\begin{align*}
    \omega_{g,n}(z_1,..,\gamma z_i,...,z_n)=\beta_{\gamma}\omega_{g,n}(z_1,...,z_i,...,z_n)
\end{align*}
for all $i\in \{1,...,n\}$. In other words, the group action on $B$ is inherited by all $\omega_{g,n}$, note that this is also true for $\omega_{0,1}$.

Let us specify the situation to our case of interest $G=\mathbb{Z}_2$, the three characters are restricted to the one dimensional representations with $\rho(\iota)^2=1$ since the group is given by $G=\{1,\iota\}$ and $\iota^2=\mathrm{Id}$, thus, the only representations are simply the \emph{trivial representation} and the \emph{sign representations}. The total number of choices that one has for a general equivariant spectral curve is $4$, since $\beta$ is uniquely determined by $\chi,\upsilon$. The symmetric bidifferential in this case has to be chosen (up to holomorphic ambiguities)
\begin{align}\label{extendedBergman}
  B(z_1,z_2)=B^{TR}(z_1,z_2)  +  \beta_\iota B^{TR}(\iota z_1,z_2)
\end{align}
to satisfy the regularity condition \eqref{bextendedregularity}. Finally, one also realizes that the $G-$orbit of a ramification point $a$ consists of ramification points again since $dx(\gamma a) = \chi_\gamma dx(a) = 0$. Thus, ramification points cannot be fixed points of the group action due to the free action around $a$. 

Before continuing, we want to remark that for both, TR and eq-TR, the correlators satisfy the linear and quadratic loop equations, this was not mentioned explicitly in \cite{Giacchetto:2021atk}, but can easily be verified. The linear, quadratic loop equation and projection property read
    \begin{itemize}
        \item \textit{Linear loop equation}: for all $g,n\geq 0$
        \begin{align}\label{LLE}
        \omega_{g,n+1}(z,z_{\llbracket n \rrbracket})+\omega_{g,n+1}(\sigma_i(z),z_{\llbracket n \rrbracket})
        \end{align}
        is holomorphic as $z\to a_i$ and has at least a simple zero in $z$ at $a_i$.

        \item \textit{Quadratic loop equation}: for all $g,n\geq 0$
        \begin{align}\label{QLE}
        \omega_{g-1,n+2}(z,\sigma_i(z),z_{\llbracket n \rrbracket})+\sum_{\substack{g_1+g_2=g\\ J_1\sqcup J_2=\llbracket n \rrbracket}}\omega_{g_1,|J_1|+1}(\sigma_i(z),z_{J_1})\omega_{g_2,|J_2|+1}(z,z_{J_2})
        \end{align}
        is holomorphic as $z\to a_i$ and has at least a double zero in $z$ at $a_i$.

        \item \textit{Projection property}: for $2g-2+n>0$, $\omega_{g,n}$ can be reconstructed from their principal parts
        \begin{align}
			\omega_{g,n}(z_1,\dots,z_n) =  \Big(\frac{1}{|G|^n} \prod_{i=1}^n \sum_{a_j}\Res_{z'_i\to a_j}	\int_{a_j}^{z'_i} B(\bullet,z_i) \Big) \omega_{g,n}(z'_1,\dots,z'_n).	
		\end{align}
    \end{itemize}
Let us also mention a slight variation for the quadratic loop equation if we have an irregular curve in the sense of \cite{DoNorbury18,DoNorbury2018,Chekhov:2017aot} which means that $y$ can have a simple pole at the ramification points of $x$. In that situation, the quadratic loop equation states that the expression \eqref{QLE} is (just) holomorphic at the ramification points (it does not necessarily have a double zero at the ramification points) \cite{Bouchard:2024fih}. On a compact curve, the projection property uniquely determines the differentials, for instance, by inheriting a normalization from $B$.

\begin{remark}
For eq-TR, the linear and quadratic loop equations are satisfied both in the setting of our \autoref{def:equTR} and in the rescaled version of \cite{Giacchetto:2021atk}, see \autoref{rmk:eqTRscaling}. This is not completely obvious at first sight, but follows from a simple degree-counting argument. Note that this is true in both cases $\omega_{0,1}=\frac{1}{|G|}ydx$ in eq-TR and $\omega_{0,1}=ydx$ in ordinary TR.
\end{remark}

\subsection{Equivariant Topological Recursion on hyperelliptic curves with
$\mathbb{Z}_2\times\mathbb{Z}_2$ symmetry}

In the different examples where eq-TR has already appeared
\cite{Giacchetto:2021atk,Alexandrov:2025kzn,fang2026involutionequivarianttopologicalrecursionmirror},
the characters were prescribed by the spectral curve. However, what happens if a
spectral curve admits more than one symmetry transformation, so that there are
several possible choices of group actions and characters? Let us consider a
spectral curve equipped with a $\mathbb{Z}_2\times\mathbb{Z}_2$ group action
(Klein four-group), arising from a global hyperelliptic curve of the form
\begin{align}
    y^2=Q(x^2).
\end{align}
The three non-trivial elements of the group act globally by
\begin{align}
    \sigma:&\, (x,y)\to (x,-y), \label{involutionss}\\
    \iota_1:&\, (x,y)\to (-x,y), \label{involutionsi1}\\
    \iota_2:&\, (x,y)\to (-x,-y), \label{involutionsi2}
\end{align}
with $\iota_2=\iota_1\sigma=\sigma\iota_1$, and $\gamma^2=1$ for all
$\gamma\in\mathbb{Z}_2\times\mathbb{Z}_2$. Each of these symmetry actions
uniquely determines the characters corresponding to eq-TR for a single
$\mathbb{Z}_2$ group. Note, however, that
$\sigma:(x,y)\to(x,-y)$ is the Galois involution of the hyperelliptic curve
and cannot be taken as the group action, since it does not act freely on the
ramification points of $x$. Thus, we are left with two different choices for
running eq-TR on this curve:
\begin{center}
\begin{tabular}{||c c c c c ||}
 \hline
 $G$ & $\phi$ & $\chi$ & $\upsilon$ & $\beta$\\ [0.5ex]
 \hline\hline
 $\mathbb{Z}_2$ & $\iota_1$ & $-1$ & $1$ & $-1$\\
 \hline
 $\mathbb{Z}_2$ & $\iota_2$ & $-1$ & $-1$ & $1$\\
 \hline
\end{tabular}
\end{center}
where the representation $1$ corresponds to the trivial representation and
$-1$ to the sign representation.

We now prove that the $\omega_{g,n}$ generated by eq-TR for these two choices
are actually the \emph{same}.

\begin{theorem}\label{thm:hyperellepticequTR}
    Consider a hyperelliptic curve of the form $y^2=Q(x^2)$. Let
    \begin{itemize}
        \item $\omega^1_{g,n}$ be the multidifferentials associated with the
        spectral curve
        $\mathcal{S}^{eq}_1=
        \{\Sigma,\iota_1,x,y,B_1,\chi_1=-1,\upsilon_1=1,\beta_1=-1\}$,
        \item $\omega^2_{g,n}$ be the multidifferentials associated with the
        spectral curve
        $\mathcal{S}^{eq}_2=
        \{\Sigma,\iota_2,x,y,B_2,\chi_2=-1,\upsilon_2=-1,\beta_2=1\}$.
    \end{itemize}
    Here, $\iota_i$ are given by
    \eqref{involutionsi1}, \eqref{involutionsi2}, and $B_i$ by
    \eqref{extendedBergman} with the corresponding $\iota_i,\beta_i$. Then,
    \begin{align*}
        \omega^1_{g,n}=\omega^2_{g,n},
        \qquad \forall\, 2g-2+n>0
    \end{align*}
    locally around any ramification point. It also holds globally on a compact curve if $B_1$ and $B_2$ are normalized with respect to the same homology basis.
    \begin{proof}
        The recursion of eq-TR is equivalent to the statement that the linear
        loop equation \eqref{LLE}, the quadratic loop equation \eqref{QLE},
        and the projection property are satisfied.

        By construction, the linear loop equation \eqref{LLE} and the projection
        property are satisfied. Thus, we only have to check that the quadratic
        loop equation is the same for both recursions. Importantly, the Bergman
        kernels $\omega_{0,2}^i=B_i$ entering the recursion and the quadratic
        loop equation are different.

       We now prove by induction on $2g-2+n$ that $\omega_{g,n}^i$ satisfy
the same quadratic loop equations. For $\omega_{1,1}^i$ there is one special term in the quadratic loop equation. Thus, for the initial condition of the induction, we first check
\begin{align*}
    \omega^1_{0,2}(\sigma(z),z)
    =&B^{TR}(\sigma(z),z)
    +\beta_1 B^{TR}(\iota_1(\sigma(z)),z)\\
    =&B^{TR}(\sigma(z),z)
    -B^{TR}(\iota_2(z),z)\\
    =&B^{TR}(\sigma(z),z)
    +B^{TR}(\iota_2(z),\sigma(z))
    +\mathcal{O}((z-a)^2)\\
    =&B^{TR}(\sigma(z),z)
    +B^{TR}(\iota_2(\sigma(z)),z)
    +\mathcal{O}((z-a)^2)\\
    =&\omega^2_{0,2}(\sigma(z),z)
    +\mathcal{O}((z-a)^2),
\end{align*}
where we have used
$B^{TR}(z_1,z)+B^{TR}(z_1,\sigma(z))
=\mathcal{O}(dx(z_1)^1dx(z)^1)$ which vanishes quadratically at the ramification points if $z_1=\iota(z)$ or $\iota_2(z)$ and $z\to a$ permuting the ramification points. This shows that the quadratic loop equations for $\omega_{1,1}^1$ and
$\omega_{1,1}^2$ coincide, and hence that
$\omega_{1,1}^1=\omega_{1,1}^2$.

          Assume now that this holds for
        $\omega_{g',n'}^i$ with
        $2g'-2+n'\leq 2g-2+n$. Then all
        $\omega_{g',n'}^1=\omega_{g',n'}^2$ already coincide, since the
        linear and quadratic loop equations together with the projection
        property uniquely determine them. In the next recursion step, at
        $2g-2+n+1$, we extract from the quadratic loop equation the terms
        that can produce a difference between the two recursions. These terms
        involve only $\omega_{0,2}^1$ or $\omega_{0,2}^2$, respectively.
        For some $z_j\in I$, one family of such terms is of the form
        \begin{align*}
            &\omega^1_{0,2}(z,z_j)
            \omega^1_{g,n}(\sigma(z),z_{I\setminus \{z_j\}})
            +\omega^1_{0,2}(\sigma(z),z_j)
            \omega^1_{g,n}(z,z_{I\setminus \{z_j\}})
            \\
            =&\bigg(
            \omega^1_{0,2}(\sigma(z),z_j)
            -\omega^1_{0,2}(z,z_j)
            \bigg)
            \omega^1_{g,n}(z,z_{I\setminus \{z_j\}})
            +\mathcal{O}((z-a)^2)
            \\
            =&\bigg(
            B^{TR}(\sigma(z),z_j)
            +\beta_1 B^{TR}(\sigma\iota_1 z,z_j)
            -B^{TR}(z,z_j)
            -\beta_1 B^{TR}(\iota_1 z,z_j)
            \bigg)
            \omega^1_{g,n}(z,z_{I\setminus \{z_j\}})
            +\mathcal{O}((z-a)^2)
            \\
            =&\bigg(
            B^{TR}(\sigma(z),z_j)
            -B^{TR}(\iota_2 z,z_j)
            -B^{TR}(z,z_j)
            +B^{TR}(\iota_2\sigma z,z_j)
            \bigg)
            \omega^1_{g,n}(z,z_{I\setminus \{z_j\}})
            +\mathcal{O}((z-a)^2)
            \\
            =&\bigg(
            \omega^2_{0,2}(\sigma(z),z_j)
            -\omega^2_{0,2}(z,z_j)
            \bigg)
            \omega^1_{g,n}(z,z_{I\setminus \{z_j\}})
            +\mathcal{O}((z-a)^2)
            \\
            =&\bigg(
            \omega^2_{0,2}(\sigma(z),z_j)
            -\omega^2_{0,2}(z,z_j)
            \bigg)
            \omega^2_{g,n}(z,z_{I\setminus \{z_j\}})
            +\mathcal{O}((z-a)^2)\\
            =&
            \omega^2_{0,2}(\sigma(z),z_j)\omega^2_{g,n}(z,z_{I\setminus \{z_j\}})
            +\omega^2_{0,2}(z,z_j)
            \omega^2_{g,n}(\sigma(z),z_{I\setminus \{z_j\}})
            +\mathcal{O}((z-a)^2).
        \end{align*}
        We have used, in the first equality, the linear loop equation for
        $\omega_{g,n}^1$; in the second, the definition of $B_1$; in the
        third, the group relations $\sigma\iota_1=\iota_2$ and
        $\iota_1=\iota_2\sigma$; in the fourth, the definition of $B_2$;
         in the fifth, the induction hypothesis; and in the last, the linear
        loop equation for $\omega_{g,n}^2$.

        This computation applies to every $z_j\in I$, and hence the quadratic
        loop equations for $\omega_{g,n+1}^1$ and $\omega_{g,n+1}^2$ coincide.
        The claim then follows from the equivalence between linear, quadratic loop equations plus projection property and the definition of eq-TR.
    \end{proof}
\end{theorem}

The theorem has several consequences. Let $\omega_{g,n}=\omega^1_{g,n}=\omega^2_{g,n}$ for $2g-2+n\neq 0$, note that $\omega_{0,2}^1\neq \omega_{0,2}^2$. Then, we have the full $\mathbb{Z}_2\times \mathbb{Z}_2$ acting on the differentials by 
\begin{align}\label{eqtromegasymmetrys}
    \omega_{g,n+1}(\sigma(z),z_I)&=-\omega_{g,n+1}(z,z_I)\\\label{eqtromegasymmetryi1}
    \omega_{g,n+1}(\iota_1(z),z_I)&=-\omega_{g,n+1}(z,z_I)\\\label{eqtromegasymmetryi2}
    \omega_{g,n+1}(\iota_2(z),z_I)&=\omega_{g,n+1}(z,z_I).
\end{align}

Let us give some examples for illustration.

\begin{example}[Rescaled and degenerate Painlev\'e II]\label{ex:RescaledPainleve2}
    The following spectral curve can be obtained, after rescaling, as a degenerate limit of the Painlev\'e II spectral curve \cite{borot2010asymptoticexpansiontracywidomgue,Iwaki:2014hmk}. It has genus~0, whereas the generic Painlev\'e II spectral curve has genus~1 \cite{Marchal:2019bia}:
    \begin{align*}
        x(z)=z+\frac{1}{z},\qquad
        y(z)=z^2-\frac{1}{z^2},\qquad
        y^2=x^2(x^2-4).
    \end{align*}
    Applications of TR to this spectral curve can be found in \cite{Iwaki:2014hmk}. Since $x$ and $y$ satisfy all the necessary symmetries, applying eq-TR gives  ($\omega_{0,1}=\frac{1}{2}ydx$ and $\omega_{0,2}=B$)
    \begin{align*}
        \omega_{0,3}(z_1,z_2,z_3)
        &=-\frac{2(z_1^2+1)(z_2^2+1)(z_3^2+1)}
        {(z_1^2-1)^2(z_2^2-1)^2(z_3^2-1)^2}
        \,dz_1\,dz_2\,dz_3,\\
        \omega_{1,1}(z)
        &=-\frac{z^2(1+z^2)}{4(z^2-1)^4}\,dz,\\
        \omega_{2,1}(z)
        &=-\frac{105z^8(z^2+1)}{16(z^2-1)^{10}}\,dz.
    \end{align*}
    We keep the list of correlators as short as possible, since the expressions become increasingly lengthy as $n$ grows. However, the examples above suffice for illustration. It is easy to check that these multidifferentials satisfy the three symmetry transformations \eqref{eqtromegasymmetrys}, \eqref{eqtromegasymmetryi1}, and \eqref{eqtromegasymmetryi2}. Thus, it does not matter which equivariant spectral curve $\mathcal{S}^1$ or $\mathcal{S}^2$ we are applying. The same multidifferentials are generated.
\end{example}

\begin{example}[Bessel system/Kummer curve]\label{ex:Kummer}
    The following spectral curve corresponds to the Bessel system of \autoref{Section2}, which was studied from the perspective of TR in \cite{Iwaki2019}:
    \begin{align*}
        x(z)=\frac{i\alpha}{2}\bigg(z+\frac{1}{z}\bigg),\qquad
        y(z)=\frac{z^2-1}{z^2+1},\qquad
        y^2=1+\frac{\alpha^2}{x^2}.
    \end{align*}
    The eq-TR multidifferentials associated with this curve are
    \begin{align*}
        \omega_{0,3}(z_1,z_2,z_3)
        &=\frac{16 i (z_1^2+1)(z_2^2+1)(z_3^2+1)}
        {\alpha (z_1^2-1)^2 (z_2^2-1)^2 (z_3^2-1)^2}
        \,dz_1\,dz_2\,dz_3,\\
        \omega_{1,1}(z)
        &=\frac{i (z^2+1)^3}
        {2\alpha (z^2-1)^4}\,dz,\\
        \omega_{2,1}(z)
        &=-\frac{i z^2 (z^2+1)^3
        \left(9 z^8+356 z^6+950 z^4+356 z^2+9\right)}
        {2\alpha^3 (z^2-1)^{10}}\,dz.
    \end{align*}
    Again, all three symmetry transformations \eqref{eqtromegasymmetrys}, \eqref{eqtromegasymmetryi1}, and \eqref{eqtromegasymmetryi2} are satisfied.
\end{example}

\subsection{Equivariant Topological Recursion on non-algebraic curves with
$\mathbb{Z}_2\times\mathbb{Z}_2$ symmetry}
The discussion of the previous subsection on hyperelliptic curves can also be applied beyond algebraic curves. Another important class of spectral curves have the form
\begin{align}\label{SWcurve}
    e^y+e^{-y}=Q(x).
\end{align}
These curves provide particular examples of Seiberg--Witten-like spectral curves.
Seiberg--Witten curves, originally introduced to describe the low-energy
effective theory of $\mathcal N=2$ supersymmetric gauge theories
\cite{Seiberg:1994rs,Seiberg:1994aj}, have subsequently played an
important role in integrable systems \cite{Gorsky:1995zq}, matrix models,
and topological string theory \cite{Klemm:2009}. Closely related spectral
curves also arise naturally in the context of TR
\cite{Norbury:2014,Dunin-Barkowski:2012kbi,FangLiuZong:2014,Borot:2019,
Borot:2025whittaker,Kazarian:2021zqh,Alexandrov:2025kzn}.

If we assume $Q(x)$ to be even $Q(x)=Q(-x)$, then the spectral curve \eqref{SWcurve} admits $\mathbb{Z}_2\times\mathbb{Z}_2$ symmetry by
\begin{align}
    \sigma:&\, (x,Y)\to (x,Y^{-1}), \label{involutionssSW1}\\
    \iota_1:&\, (x,Y)\to (-x,Y), \label{involutionsi1SW1}\\
    \iota_2:&\, (x,Y)\to (-x,Y^{-1}). \label{involutionsi2SW2}
\end{align}
If we assume $Q(x)$ to be odd $Q(x)=-Q(-x)$, then the spectral curve \eqref{SWcurve} admits $\mathbb{Z}_2\times\mathbb{Z}_2$ symmetry by
\begin{align}
    \sigma:&\, (x,Y)\to (x,Y^{-1}), \label{involutionssSW}\\
    \iota_1:&\, (x,Y)\to (-x,-Y), \label{involutionsi1SW}\\
    \iota_2:&\, (x,Y)\to (-x,-Y^{-1}), \label{involutionsi2SW}
\end{align}
with $Y=e^y$, thus $\sigma$ is identified (in both cases) globally with the Galois involution sending $y\to -y$. The factor of $-1$ in the second argument (in front of $Y$ or $Y^{-1}$) in the second and third involution, does not imply anything to $\beta$ since $y=\log Y$ and this trivial sign shifts $y$ which does not change the generated $\omega_{g,n}$ by either TR or eq-TR. This directly implies that the statement of \autoref{thm:hyperellepticequTR} also applies to the spectral curves of \eqref{SWcurve} in the special cases of even (or odd $Q$) with involutions \eqref{involutionsi1SW1} and \eqref{involutionsi2SW2} (or with involutions \eqref{involutionsi1SW} and \eqref{involutionsi2SW}, respectively).
\begin{corollary}
    Consider a curve of the form $e^y+e^{-y}=Q(x)$ with $Q$ even (or odd, respectively). Let
    \begin{itemize}
        \item $\omega^1_{g,n}$ be the multidifferentials associated with the
        spectral curve
        $\mathcal{S}^{eq}_1=
        \{\Sigma,\iota_1,x,y,B_1,\chi_1=-1,\upsilon_1=1,\beta_1=-1\}$,
        \item $\omega^2_{g,n}$ be the multidifferentials associated with the
        spectral curve
        $\mathcal{S}^{eq}_2=
        \{\Sigma,\iota_2,x,y,B_2,\chi_2=-1,\upsilon_2=-1,\beta_2=1\}$.
    \end{itemize}
    Here, $\iota_i$ are given by \eqref{involutionsi1SW1}, \eqref{involutionsi2SW2} for $Q$ even, and $B_i$ by
    \eqref{extendedBergman} with the corresponding $\iota_i,\beta_i$
    (or \eqref{involutionsi1SW}, \eqref{involutionsi2SW} for $Q$ odd, and $B_i$ by
    \eqref{extendedBergman} with the corresponding $\iota_i,\beta_i$, respectively). Then,
    \begin{align*}
        \omega^1_{g,n}=\omega^2_{g,n},
        \qquad \forall\, 2g-2+n>0.
    \end{align*}
\end{corollary}

Let us also illustrate a class of examples involving non-algebraic curves.

\begin{example}[Chebyshev curves]\label{ex:Chebyshev}
    Let $T_m(x)$ be the Chebyshev polynomial satisfying, for $x=z+\frac{1}{z}$, the relation
        $z^m+z^{-m}=2T_m\left(\frac{x(z)}{2}\right)$.
    We then consider the spectral curve
    \begin{align*}
        x(z)=z+\frac{1}{z},\qquad 
        y(z)=m\log z,\qquad 
        e^y+e^{-y}=2T_m\left(\frac{x}{2}\right).
    \end{align*}
    For $m=1$, this curve was studied from the TR perspective in \cite{Norbury:2014} in connection with Gromov--Witten invariants of $\mathbb{P}^1$, and the conjectural relation was proved in \cite{Dunin-Barkowski:2012kbi}. This spectral curve was also studied from the eq-TR perspective in \cite{Alexandrov:2025kzn}, providing a new spin on $\kappa$-classes. The first few eq-TR multidifferentials are
    \begin{align*}
        \omega_{0,3}(z_1,z_2,z_3)
        &=-
        \frac{8(z_1^2+1)(z_2^2+1)(z_3^2+1)}
        {m(z_1^2-1)^2(z_2^2-1)^2(z_3^2-1)^2}
        \,dz_1\,dz_2\,dz_3,\\[2mm]
        \omega_{1,1}(z)
        &=-
        \frac{(z^2+1)(z^4+4z^2+1)}
        {6m(z^2-1)^4}
        \,dz,\\[2mm]
        \omega_{2,1}(z)
        &=
        \frac{z^2(z^2+1)(z^4+8z^2+1)(z^8-144z^6-344z^4-144z^2+1)}
        {15m^3(z^2-1)^{10}}
        \,dz.
    \end{align*}
    These clearly satisfy the symmetry transformations \eqref{eqtromegasymmetrys}, \eqref{eqtromegasymmetryi1}, and \eqref{eqtromegasymmetryi2}.
\end{example}
After some examples illustrating the invariance of the differentials on the choice of the symmetry, we are now ready to proceed to the folding of meromorphic differentials. This constitutes the TR analogue of the folding operation on the meromorphic differential system studied in \autoref{Section2}.

\subsection{Folding equivariant topological recursion with $\mathbb{Z}_2\times\mathbb{Z}_2$ to ordinary topological recursion}\label{sec:foldingomegagn}

Let us now restrict ourselves to the setting of a genus 0 spectral curve. The three involutions of the $\mathbb{Z}_2\times\mathbb{Z}_2$-group can be chosen on $\mathbb{P}^1$ as
\begin{align*}
    \sigma:z &\to \frac{1}{z}, \qquad
    \iota_1:z \to -z, \qquad
    \iota_2:z \to -\frac{1}{z}.
\end{align*}
This can be achieved by a Möbius transformation and the fact that the involutions are defined globally. Furthermore, having a function $x$ on $\mathbb{P}^1$ means that it is rational in a global coordinate $z$. If we restrict to a rational function $x$ satisfying the involution identities \eqref{involutionss}, \eqref{involutionsi1}, and \eqref{involutionsi2}, and having only two simple ramification points, it has to have the following form
\begin{align}\label{xrational}
    x(z)=c\bigg(z+\frac{1}{z}\bigg)^{\pm 1}.
\end{align}
The rational function $y$ together with the desired symmetry transformations has, in general, the representation
\begin{align}\label{yrational}
    y(z)=\bigg(z+\frac{1}{z}\bigg)\bigg(z-\frac{1}{z}\bigg)R_y\bigg(z^2+\frac{1}{z^2}\bigg),
\end{align}
where again $R_y$ is some rational function invariant under all involutions, and $\big(z+\frac{1}{z}\big)\big(z-\frac{1}{z}\big)$ transforms as desired.

\begin{remark}\label{rmk:ywithlog}
    Allowing more general $y$, as in the case of a non-algebraic curve, means allowing $dy$ to be meromorphic, which is rational for genus 0 curves. The transformations under the involutions should then hold up to an additive constant. This means, for instance, that $y$ can contain a logarithm $\log z$, which transforms in the correct way. Thus, we also want to consider
    \begin{align}
        y(z)=c_0\log z+y_{\mathrm{rational}}(z),
    \end{align}
    where $y_{\mathrm{rational}}$ is given by \eqref{yrational}. It is easy to check that this family satisfies the desired properties under the involution transformations up to an additive constant. We further assume that $x$ has at least a simple pole at $z=0,\infty$; otherwise, the use of Log-TR \cite{Alexandrov:2023tgl} seems necessary, which we want to avoid. In principle, one can allow more logarithmic terms transforming accordingly under the involutions; however, for the purposes of this article, we restrict ourselves to the form above.
\end{remark}

The extended Bergman kernel $B_2$ defined in \autoref{thm:hyperellepticequTR} (with character $\beta=1$), used for eq-TR, has in this coordinate the form
\begin{align}\label{Bergmanfolding}
    B_2(z_1,z_2)
    =\frac{dz_1\,dz_2}{(z_1-z_2)^2}+\frac{dz_1\,dz_2}{(1+z_1z_2)^2}
    =\frac{dw(z_1)dw(z_2)}{(w(z_1)-w(z_2))^2}
    =B^{TR}(w_1,w_2),
\end{align}
where $w(z)=z-\frac{1}{z}=z+\iota_2(z)$ and $w_i=w(z_i)$. Thus, the extended Bergman kernel of eq-TR with character $\beta=1$ and involution $\iota_2$ descends, via the cover $w(z)=z-\frac{1}{z}=z+\iota_2(z)$, to the ordinary TR Bergman kernel $B^{TR}(w_1,w_2)$. This can be seen as a ``folding'' of $B_2$ on the $z$-space to $B^{TR}$ on the $w$-space in the following sense.

\begin{definition}\label{def:foldingdifferentials}
Let $\omega_n$ be a rational multidifferential on $(\mathbb{P}^1_z)^n$ which is invariant under the pullback of an involution $\iota$, that is, it satisfies
\begin{align*}
  \iota^*_i \omega_{n}(z_1,\ldots,z_n)  := \omega_{n}(z_1,\ldots,\iota z_i,\ldots,z_n)
    =\omega_{n}(z_1,\ldots,z_i,\ldots,z_n).
\end{align*}
For $w:\mathbb{P}^1_z\to \mathbb{P}^1_w$ with $w(z)=z+\iota z$, there exists a rational multidifferential $\overline{\omega_n}(w_1,\ldots,w_n)$ on $(\mathbb{P}^1_w)^n$, with $w_i=w(z_i)$ coinciding with $\omega_{n}$. The unique multidifferential $\overline{\omega_n}$ is called the \emph{folding} of $\omega_n$ under $\iota$.
\end{definition}
This operation is usually referred to as a \emph{descent}, but to keep the terminology consistent with \autoref{Section2}, we will call it a folding. The folding is unique, since it is fully determined by the pole structure since we work on $\mathbb{P}^1$. Note that one discussing a folding of a meromorphic differential from one space to another, the function relating the two spaces (here $w(z)$) is taken to be non-constant, this is evidently the case for our involution $\iota_2$.

We have already seen in \eqref{Bergmanfolding} a folding of $B_2$ under $\iota_2$. The reason is that $\iota_2$ transforms $B_2$ with character $\beta_2=1$. However, $B_1$ of \autoref{thm:hyperellepticequTR} cannot be written as a folding under $\iota_1$ of a Bergman kernel of ordinary TR, since $\beta_1=-1$. Let us give another example of the folding of a (multi-)differential.

\begin{example}
    Let $x,y$ be as in \eqref{xrational} and \eqref{yrational}. The differential $\omega^{eq}_{0,1}=\frac{1}{2}y(z)dx(z)$ has a folding under $\iota_2$, since it transforms accordingly with respect to $\iota_2$. However, $x,y$ have no foldings as rational functions. Nevertheless, the folding $\overline{\omega_{0,1}}$ can be constructed by
    \begin{align*}
        \overline{\omega^{eq}_{0,1}}(w(z))
        =\frac{1}{2}\frac{y(z)}{2x(z)}d\big(x(z)^2\big)
        =\frac{1}{2}y_{new}(w)dx_{new}(w)=\frac{1}{2}\omega^{TR}_{0,1}(z).
    \end{align*}
    Thus, replacing $y\to\frac{y(z)}{2x(z)}=y_{\mathrm{new}}(w)$ and $x\to x(z)^2=x_{\mathrm{new}}(w)$, with $x,y$ given by \eqref{xrational} and \eqref{yrational}, provides global rational functions in $w$, since $z^2+\frac{1}{z^2}=w^2+2$.
\end{example}

\begin{remark}
    Note that, if we include logarithms for $y$ as in \autoref{rmk:ywithlog}, the corresponding $y_{\mathrm{new}}$ is not globally defined, just as $y$ was not. Moreover, the differential $dy_{\mathrm{new}}=d\big(\frac{y(z)}{2x(z)}\big)$ is not meromorphic, neither in $z$ nor in $w$. However, focusing on $z=\pm1$, corresponding to $w=0$, the two local inverse branches satisfy $z_-(w)=-z_+(-w)$. Together with the involution properties of $x$ and $y$, this implies that $y_{\mathrm{new}}$ admits a local odd expansion around $w=0$, i.e. $y_{\mathrm{new}}(-w)=-y_{\mathrm{new}}(w).$ For our later purposes, this local property will be sufficient.
\end{remark}

Since the multidifferentials generated by eq-TR with $x,y$ given by \eqref{xrational} and \eqref{yrational} on the spectral curve $\mathcal{S}^{eq}_2=\{\Sigma,\iota_2,x,y,B_2,\chi_2=-1,\upsilon_2=-1,\beta_2=1\}$ (or, due to \autoref{thm:hyperellepticequTR}, equivalently by $\mathcal{S}^{eq}_1$) are invariant under the involution $\iota_2$, these differentials $\omega_{g,n}$ naturally have a folding as defined in \autoref{def:foldingdifferentials}. We prove that the folding $\overline{\omega_{g,n}}$ is actually given by ordinary TR on the spectral curve $\mathcal{S}^f=(\Sigma^f,x^2,\frac{y}{2x},B^{TR})$ with $w:\Sigma\to\Sigma^f$.

\begin{theorem}\label{thm:foldingomegagn}
    Let $\omega^{eq}_{g,n}$ be generated by eq-TR on the spectral curve $\mathcal{S}^{eq}_2$, and let $\omega^{TR}_{g,n}$ be generated by TR on the spectral curve $\mathcal{S}^f$ with $w:\Sigma\to\Sigma^f$ and $w(z)=z+\iota_2 z$, as above. Then the folding $\overline{\omega^{eq}_{g,n}}$ of $\omega^{eq}_{g,n}$ under $\iota_2$ to $\Sigma^f$ coincides with $\omega^{TR}_{g,n}$ up to a symmetry factor; more precisely,
    \begin{align}
        \overline{\omega^{eq}_{g,n}}(w_1,\ldots,w_n)
        =2^{2g+n-2}\omega^{TR}_{g,n}(w_1,\ldots,w_n).
    \end{align}
    Equivalently, on $\Sigma$, the following identification holds:
    \begin{align}
        \omega^{eq}_{g,n}(z_1,\ldots,z_n)
        =2^{2g+n-2}\omega^{TR}_{g,n}(w(z_1),\ldots,w(z_n)).
    \end{align}

\begin{proof}
    The proof follows an idea similar to that of \cite[Prop.~3.3]{Alexandrov:2025kzn}. As discussed before, we already have the identifications $\overline{\omega^{eq}_{0,1}}(w)=\frac{1}{2}\omega^{TR}_{0,1}(w)$ and $\overline{\omega^{eq}_{0,2}}(w_1,w_2)=\omega^{TR}_{0,2}(w_1,w_2)$. Due to the identification in the $(g,n)=(0,1),(0,2)$ sectors, the recursions themselves can be identified. The ramification points of $x$ in $z$ and of $x^2$ in $w$ coincide only locally in the recursion of TR on $\mathcal{S}^f$. Thus, for the actual comparison, one has to choose a branch
    \begin{align*}
        z_\pm(w)=\frac{w\pm\sqrt{4+w^2}}{2},
    \end{align*}
    where the residue computation for $z_+$ reproduces the residue at $z=1$, while that for $z_-$ reproduces the residue at $z=-1$. However, since a branch needs to be chosen, only one of them is generated at each recursive step. The analytic continuation of $\omega_{g,n}^{eq}$ by $\iota_2$ yields at each recursive step a factor of $\frac{1}{2}$ through
    \begin{align*}
        \underset{z=1}{\Res}
        =\frac{1}{2}\bigg(
        \underset{z=1}{\Res}
        +
        \underset{z=-1}{\Res}
        \bigg),
    \end{align*}
    where the equality of the two residue contributions follows from the $\iota_2$-invariance. In this way, the full formula for eq-TR is reproduced together with a factor of $\frac{1}{2}$ for each $\omega_{g,n}^{eq}$. Finally, by an induction argument in $2g-2+n$ (with the initial induction data being satisfied), one obtains
    \begin{align*}
        2^{2-2g-n}\omega^{eq}_{g,n}(z_1,\ldots,z_n)
        =
        \omega^{TR}_{g,n}(w(z_1),\ldots,w(z_n)),
    \end{align*}
    which proves the assertion.

    The step from $\omega^{eq}_{g,n}(z_1,\ldots,z_n)$ to $\overline{\omega^{eq}_{g,n}}(w_1,\ldots,w_n)$, as defined in \autoref{def:foldingdifferentials}, is unique since $\omega^{eq}_{g,n}$ is invariant under $\iota_2$.
\end{proof}
\end{theorem}  

The theorem can also be inverted. So in principle, if we have spectral curve with $x=w^2+4$ in some global coordinate on $\mathbb{P}^1$ (the shift by 4 is irrelevant for TR), and $y$ odd under reflection $w\to-w$, then the constructed multidifferentials can also be computed from an eq-TR spectral curve. In some sense, the multidifferentials for those curves can be \emph{unfolded}.

We want to pick up on the previous examples and discuss their foldings:

\begin{example}[Folding rescaled and degenerate Painlev\'e II to Airy]\label{ex:RescaledPainleve2Fold}
    The rescaled and degenerate Painlev\'e II curve of \autoref{ex:RescaledPainleve2} can be folded to the Airy curve. Following the folding procedure, we find the folded spectral curve
    \begin{align*}
        \mathcal{S}^f=(\mathbb{P}^1,x^f,y^f,B^{TR})
    \end{align*}
    with $x^f(w(z))=x(z)^2=w(z)^2+4$ and $y^f(w(z))=\frac{y(z)}{2x(z)}=\frac{w(z)}{2}$, where $x,y$ are given in \autoref{ex:RescaledPainleve2}. This curve corresponds exactly to the shifted Airy curve
    \begin{align*}
        (y^f)^2+1-\frac{x^f}{4}=0,
    \end{align*}
    which encodes $\psi$-class intersection numbers \cite{Eynard:2007kz}. The shift does not have any impact on the correlators or the enumerative invariants. Thus, \emph{the folding of the equivariant rescaled and degenerate Painlev\'e II spectral curve is the Airy spectral curve}.
\end{example}

\begin{example}[Folding Kummer to Bessel (Bessel system to Bessel--Clifford system)]\label{ex:KummerFold}
    The Kummer curve of \autoref{ex:Kummer} can also be folded to the spectral curve
    \begin{align*}
        \mathcal{S}^f=(\mathbb{P}^1,x^f,y^f,B^{TR})
    \end{align*}
    with $x^f(w(z))=x(z)^2=-\frac{\alpha^2}{4}(w(z)^2+4)$ and
    \[
        y^f(w(z))=\frac{y(z)}{2x(z)}
        =-\frac{iw(z)}{\alpha(4+w(z)^2)},
    \]
    where $x,y$ are given in \autoref{ex:Kummer}. The spectral curve can equivalently be written as
    \begin{align*}
        (y^f)^2-\frac{\alpha^2}{4(x^f)^2}-\frac{1}{4x^f}=0.
    \end{align*}
    This curve is called the \emph{Bessel curve} in \cite{Iwaki2019}. Thus, \emph{the folding of the equivariant Kummer spectral curve is the Bessel spectral curve}. In terms of the isomonodromic systems discussed in \autoref{Section2}, this statement connects the Bessel system with the Bessel--Clifford system, as will be discussed in the next subsection.
\end{example}

\begin{example}[Folding of the Chebyshev curve]
    The Chebyshev curve of \autoref{ex:Chebyshev} can be folded to the spectral curve
    \begin{align*}
        \mathcal{S}^f=(\mathbb{P}^1,x^f,y^f,B^{TR})
    \end{align*}
    with
    \[
        x^f(w(z))=x(z)^2=w(z)^2+4
    \]
    and
    \[
        y^f(w(z))
        =\frac{y(z)}{2x(z)}
        =m\frac{\log z}{2(z+\frac{1}{z})}
        =m\frac{\operatorname{arsinh}\!\left(\frac{w}{2}\right)}
        {2\sqrt{w^2+4}},
    \]
    where $x,y$ are given in \autoref{ex:Chebyshev}. Note that $y$ can be shifted by a constant. For the two different solutions $z_\pm$ corresponding to a given $w$, one may therefore choose $m\log(\pm z)$, yielding the same local expansion at the ramification points. Thus, the corresponding $y^f$ has a unique local expansion at $w=0$.

    This spectral curve already appeared in \cite{Alexandrov:2025kzn}, where it was used to derive a new spin on polynomial relations among $\kappa$-classes. We want to emphasise that $dy^f$ is not globally meromorphic, and the relation between $x^f$ and $y^f$ reads
   \begin{align*}
    2T_{m/2}\!\left(\frac{x^f-2}{2}\right)
    =e^{2y^f\sqrt{x^f}}
    +e^{-2y^f\sqrt{x^f}}.
\end{align*}
    This spectral curve is rather unusual, but TR can still be performed provided that $\infty$ is neglected for the residue calculation.
\end{example}

\subsection{TR-folding at the level of quantum curves}

We want to apply the discussed folding between eq-TR and ordinary TR at the level of multi-differentials to the quantum curve and the associated wave function. It was observed in TR that the generating function of the integrated multi-differentials asymptotically solves an ODE which can be seen as a quantum curve of the original spectral curve \cite{Gukov:2011qp,Norbury:2015lcn,Bouchard:2016obz}. For higher-genus spectral curves, some non-perturbative effects need to be added \cite{Iwaki-P1,Marchal:2019bia,Eynard:2023fil}; the final result regarding higher-genus algebraic spectral curves was formulated in \cite{Eynard:2021sxg}. To our knowledge, this provides the first explicit derivation of quantum curves associated with the equivariant TR considered here, obtained through an explicit folding relation with ordinary TR.

\medskip

In our context, where we can fold an equivariant spectral curve and its multi-differentials, we can actually identify its quantum curve from the quantum curve of the folded TR spectral curve by applying \autoref{thm:foldingomegagn}. Let us be more precise about the actual structure. For genus-zero spectral curves, let us define the wave function from a given system of multi-differentials $\omega_{g,n}$ by
\begin{align*}
    \psi_b(z;\hbar)=\exp \sum_{n,g}\frac{\hbar^{2g-2+n}}{n!}\int_b^z...\int_b^z\bigg(\omega_{g,n}-\delta_{g,0}\delta_{n,2}\frac{dx(z_1)dx(z_2)}{(x(z_1)-x(z_2))^2}\bigg)
\end{align*}
where the sum over $n$ starts at 1 and the sum over $g$ at 0. Since the wave function depends on the base point $b$, it is very common to choose $b$ such that $x(b)$ is singular. This simplifies the final quantum curve, that is, the ODE annihilating the wave function \cite{Bouchard:2016obz}, but the two unstable coefficients generally need some regularization, which does not affect the final outcome in an important way.

Let us first show that, for the equivariant multi-differentials considered in \autoref{sec:foldingomegagn}, there is no difference between choosing $\psi_0$ and $\psi_\infty$ for the stable multi-differentials, up to a normalization factor.
\begin{lemma}\label{lem:basepoint}
    Let $\omega_{g,n}$ be generated by equivariant TR for a genus-zero curve with a $\mathbb{Z}_2\times \mathbb{Z}_2$ involution. Then, with the coordinate choice in \autoref{sec:foldingomegagn}, we have
    \begin{align}
        \int_0^\infty\omega_{g,n+1}(\bullet,I)=0
    \end{align}
    for all $2g+n-2\geq0$ and the integral contour avoids the poles. In particular, this implies that, for appropriately regularized wave functions constructed from $\omega_{g,n}$, we have $\psi_0(z;\hbar)=C\psi_\infty(z;\hbar)$.
    \begin{proof}
        We have the identity $\omega_{g,n+1}(\iota_2 z,I)=\omega_{g,n+1}( z,I)$ with $\iota_2z=-\frac{1}{z}$; see \eqref{eqtromegasymmetryi2}. Thus,
        \begin{align*}
            \int_0^\infty\omega_{g,n+1}(\bullet,I)=\int_0^\infty\omega_{g,n+1}(\iota_2\bullet,I)=\int_{-\infty}^0\omega_{g,n+1}(z,I)=-\int^{-\infty}_0\omega_{g,n+1}(z,I).
        \end{align*}
        Since also $\int_{-\infty}^\infty\omega_{g,n+1}(\bullet,I)=0$, the identity follows.

        The statement about the wave function follows by rewriting any integral
        \begin{align*}
            \int_0^z\omega_{g,n+1}(\bullet,I)=\bigg(\int_0^\infty+\int_\infty^z\bigg)\omega_{g,n+1}(\bullet,I)=\int_\infty^z\omega_{g,n+1}(\bullet,I)
        \end{align*}
        appearing in the construction of the wave function.
    \end{proof}
\end{lemma}
This means that a wave function coming from equivariant TR with the given symmetry is independent of the singular base point. From this, we can finally deduce an explicit relation between the wave function generated by eq-TR and the wave function of the folded differentials from ordinary TR:
\begin{proposition}\label{prop:wave}
    Let $\psi^{eq}_{0/\infty}(z;\hbar)$ be the regularized wave function constructed from $\omega^{eq}_{g,n}$ by eq-TR on the spectral curve $\mathcal{S}^{eq}_2$, and let $\psi^{TR}_{\infty}(w;\hbar)$ be the regularized wave function constructed from $\omega^{TR}_{g,n}$ by TR on the spectral curve $\mathcal{S}^f$ with $w:\Sigma\to\Sigma^f$ and $w(z)=z+\iota_2 z$. Then, we have
    \begin{align}
        \psi^{eq}_{0/\infty}(z;\hbar)=C \sqrt{x(z)}\psi^{TR}_\infty(w(z);2\hbar).
    \end{align}
    \begin{proof}
        By \autoref{lem:basepoint}, the base point $b$ can be chosen to be either 0 or $\infty$. By \autoref{thm:foldingomegagn}, all $\omega_{g,n}^{eq}$ and $\omega_{g,n}^{TR}$ are identified up to a symmetry factor that amounts to multiplying $\hbar$ by 2. Lastly, for $\psi^{TR}$, the regularization of $\omega_{0,2}^{TR}$ is performed using $x_{new}=x^2$. Thus, the two regularizations of the corresponding $\omega_{0,2}^{eq/TR}$ are different. One uses $x$ and the other uses $x_{new}=x^2$, giving rise to
        \begin{align*}
            -\frac{1}{2}\int^x\int^x_{reg}\frac{dx_1dx_2}{(x_1-x_2)^2}-\frac{d(x_1^2)d(x_2^2)}{(x_1^2-x_2^2)^2}=-\frac{1}{2}\int^x\int^x_{reg}\frac{dx_1dx_2}{(x_1+x_2)^2}=\frac{1}{2}\log 2x,
        \end{align*}
        which, upon exponentiation, gives rise to $\tilde{C}\sqrt{x}$.
    \end{proof}
\end{proposition}
From the proposition, we can finally conclude that, if we have a quantum curve annihilating the wave function constructed by ordinary TR on the folded spectral curve $\mathcal{S}^f$, we can derive a quantum curve annihilating the wave function coming from eq-TR on the corresponding spectral curve $\mathcal{S}$. The explicit relation between the quantum curves is formulated in the following theorem:
\begin{theorem}\label{thm:quantumcurve}
    Let $\psi^{eq}_{0/\infty} ( z;\hbar)$ be the regularized wave function constructed from $\omega^{eq}_{g,n}$ by eq-TR on the spectral curve $\mathcal{S}^{eq}_2$, and let $\psi^{TR}_{\infty}(w;\hbar)$ be the regularized wave function constructed from $\omega^{TR}_{g,n}$ by TR on the spectral curve $\mathcal{S}^f$ with $w:\Sigma\to\Sigma^f$ and $w(z)=z+\iota_2 z$. Further, let $\hat{P}^{eq}_{\hbar}(\hat{x},\hat{y})$ be the quantum curve annihilating $\psi^{eq}_\infty ( z;\hbar)$ with $\hat{x}=x(z)$ and $\hat{y}=\hbar\frac{\partial}{\partial x(z)}$, and let $\hat{P}^{TR}_{\hbar}(\hat{x}_{new},\hat{y}_{new})$ be the quantum curve annihilating $\psi^{TR}_\infty ( w;\hbar)$ with $\hat{x}_{new}=x_{new}(w)$ and $\hat{y}_{new}=\hbar\frac{\partial}{\partial x_{new}(w)}$, where $x,y,x_{new},y_{new}$ are as in \autoref{sec:foldingomegagn}. Then, the two quantum curves are related by
    \begin{align}\label{eq:thmquantumcurve}
        \hat{P}^{eq}_{\hbar/2}(\hat{x},\hat{y})=\hat{P}^{TR}_\hbar \bigg(\hat{x}^2,\frac{1}{2\hat{x}}\hat{y}-\frac{\hbar}{4\hat{x}^2}\bigg).
    \end{align}
\end{theorem}
The theorem provides a class of examples where it is now known that eq-TR gives rise to a quantum curve which is related to folding at the level of differentials. To our knowledge, this is the first time quantum curves have been computed for equivariant TR, although this is achieved by relating it back to ordinary TR.

Both wave functions satisfy second-order differential equations, each of which has two independent solutions. At the formal WKB level, these solutions are characterised by $\pm\hbar$ or by a choice of square-root branch when passing from $z$ to $x$, which is essentially reflected by $\hbar\to-\hbar$. Therefore, by abuse of notation, we can write $\psi^{eq}_{0/\infty}(x;\hbar)$ and $\psi^{TR}_\infty(x_{new};\hbar)$. In these terms, it follows from \autoref{prop:wave} that the relation between these wave functions is
\begin{align*}
    \psi^{eq}_{0/\infty}(x;\hbar)
    =
    C\sqrt{x}\,\psi^{TR}_\infty(x^2;2\hbar).
\end{align*}
Let us illustrate the statements above using some of the examples that appeared earlier.

\begin{example}[Quantum curve of the equivariant rescaled and degenerate Painlev\'e II curve]
    This example is a continuation of \autoref{ex:RescaledPainleve2} and \autoref{ex:RescaledPainleve2Fold}. The folded differentials correspond to the Airy spectral curve, and the quantum curve annihilating $\psi^{TR}_\infty(x_{new};\hbar)$ was computed in \cite{Bergere:2009zm,Bouchard:2016obz}. The ordinary TR quantum curve reads
    \begin{align*}
        \hat{P}^{TR}_\hbar(\hat{x}_{new},\hat{y}_{new})
        =
        \hat{y}_{new}^2
        -
        \frac{\hat{x}_{new}-4}{4}.
    \end{align*}
    We can conclude from \autoref{thm:quantumcurve} that the quantum curve annihilating the equivariant wave function $\psi^{eq}_0(x;\hbar)$ is of the form (after multiplication by $\hat{x}^2$)
    \begin{align*}
        \hat{P}^{eq}_\hbar(\hat{x},\hat{y})
        =
        4\hat{y}^2
        -
        8\frac{\hbar}{\hat{x}}\hat{y}
        +
        5\frac{\hbar^2}{\hat{x}^2}
        -
        \hat{x}^2(\hat{x}^2-4),
    \end{align*}
    which reduces to the classical curve of \autoref{ex:RescaledPainleve2}, with the rescaling $y\to 2y$ required by $\omega_{0,1}=\frac{1}{2}y\,dx$ for eq-TR.
\end{example}

\begin{example}[Quantum curve of the equivariant Kummer curve]
    This example is a continuation of \autoref{ex:Kummer} and \autoref{ex:KummerFold}. The quantum curve was derived in \cite{Iwaki2019} from variational formulas and is called the Bessel curve in \emph{op.\ cit.}, but is given with a different parametrization with $\beta_{0+}=i,\, \beta_{0-}=-i$. The monodromy at $\{ 0 \}$ are related by $4\lambda_0^2=(\alpha+\frac{1}{2}) ^2-1$ (we will denote it under by $\alpha_0$ for simplicity), and we take the integration divisor there with $\nu_{0+}=\nu_{0-}=0$. Thus, the quantum curve reads
    \begin{align*}
        \hat{P}^{TR}_\hbar(\hat{x}_{new},\hat{y}_{new})
        =
        \hat{y}_{new}^2
        +
        \frac{\hbar}{\hat{x}_{new}}\hat{y}_{new}
        -
        \frac{\hat{x}_{new}+\alpha_0^2}{4\hat{x}_{new}^2}.
    \end{align*}
    Note that this curve coincides completely with the curve given in \eqref{Bessel-Cliffordcurvequantum}. We can conclude from \autoref{thm:quantumcurve} that the quantum curve annihilating the equivariant wave function $\psi^{eq}_0(x;\hbar)$ is of the form (after multiplication by $\hat{x}^2$)
    \begin{align*}
        \hat{P}^{eq}_\hbar(\hat{x},\hat{y})
        =
        4\hat{y}^2
        -
        1
        -
        \frac{\alpha_0^2-\hbar^2}{\hat{x}^2},
    \end{align*}
    which reduces to the classical curve of \autoref{ex:Kummer}, with the rescaling $y\to 2y$ required by $\omega_{0,1}=\frac{1}{2}y\,dx$ for eq-TR. Note that this curve matches the quantum curve produced by the Bessel system of the Lax matrix in \eqref{eq2-30} and, up to an additional scaling factor, the wave functions can be explicitly given in terms of Bessel functions.
\end{example}

\section{Further directions} \label{sec4}

The results obtained above suggest that the folding relating the Bessel and Bessel--Clifford systems should be viewed as part of a broader geometric mechanism rather than as a phenomenon specific to a single pair of special functions. In particular, the compatibility between the symmetry reduction of meromorphic connections, the passage to a quotient spectral curve, and the corresponding transformation of topological recursion data raises several natural questions.

\begin{itemize}
    \item The classical hypergeometric differential equations are organized by a diagram of confluences in which successive coalescences of singularities relate the Gauss hypergeometric equation to the confluent hypergeometric, Bessel, and related equations. The analysis in this paper suggests that this diagram may admit a geometric refinement in terms of symmetry reduction and folding of meromorphic connections. Our construction realizes, at the level of connections, the passage
\begin{align}
\mathrm{Bessel}
\longrightarrow
\mathrm{Bessel\mbox{-}Clifford}
\end{align}
through a $\mathbb Z_2$-equivariant structure and the quotient coordinate $\xi=\lambda^2$. At the same time, the Schrödinger-gauge description identifies the corresponding potentials with the Kummer and Bessel potentials appearing in the confluent hypergeometric picture. Thus the usual relation between special-function equations is accompanied by a relation between their underlying meromorphic connections, spectral curves, and formal types. This raises the question of whether the “entire hypergeometric confluence diagram can be lifted to a diagram of equivariant meromorphic connections", in which arrows corresponding to confluence are supplemented by arrows corresponding to finite coverings and foldings. One may equivalently ask whether there exists a categorical or moduli-theoretic formulation of these results. The parameter shift $\alpha \longmapsto \alpha-\frac{1}{2}$ appearing in the present example would then be interpreted as part of a general rule governing the transformation of formal monodromy under ramified descent. A particularly interesting problem is to determine whether the fractional shifts of monodromy parameters produced by folding can be described directly from the ramification indices and the characters of the equivariant structure. Such a result would provide a conceptual explanation for the shifts.
\item The present construction is based on a $\mathbb Z_2$-symmetry. One natural direction is studying folding for general finite groups acting on meromorphic connections and their spectral data (Dihedral groups, higher finite groups, etc). 
One natural question is whether successive foldings can be performed consistently when studying for instance two different group actions
\begin{align}
    \Sigma \longrightarrow \Sigma/G_1 \longrightarrow (\Sigma/G_1)/G_2 \cong
    \Sigma/(G_1\times G_2), \qquad\text{if the two actions commute},
\end{align}
and whether the resulting connection and spectral data depend only on the subgroup generated by $G_1$ and $G_2$. It would also be interesting to determine which finite-group actions on rank-two connections lead, after quotienting, to ramified rank-two systems of special-function type. This could provide a systematic mechanism for producing new special-function equations from invariant sectors of larger isomonodromic systems. On this point, the higher rank isomonodromy systems offer a wide range of possibilities and freedom to study different group reductions and foldings.
\item  The relation $\xi=\lambda^2,$ with $\widetilde y=2\lambda y$ identifies the spectral differentials on the unfolded and folded curves $\tilde{y}\,d\lambda=y\,d\xi$. This indicates that the natural object preserved by the folding is not necessarily the individual coordinate functions $x$ and $y$, but rather the spectral differential. This observation suggests investigating folding intrinsically at the level of spectral curves equipped with their meromorphic differentials. One may ask under which conditions a $G$-equivariant spectral curve $(\widetilde\Sigma,\widetilde x,\widetilde y,B)$ admits a quotient spectral curve such that $\pi^*(y\,dx)=\widetilde y\,d\widetilde x$. Such a formulation would separate the geometric content of the folding from the particular choice of coordinates and gauges.

The example considered here also shows that the Bergman kernel requires special attention. The pullback of the Bergman kernel from the quotient is not simply the ordinary Bergman kernel on the covering curve; rather, it contains an orbit contribution. This suggests that the appropriate notion of a folded spectral curve should include, as part of its data, a prescription for the transformation of the bidifferential. The example studied suggests that eq-TR may provide the natural recursive counterpart of geometric folding. In the present example, the modified bidifferential $B\neq B^{TR}$
is precisely the object which compensates for the difference between the bidifferential on the covering curve and the pullback of the bidifferential on the quotient. A natural problem is therefore to formulate a general “folding theorem for TR". Given a finite-group-equivariant spectral curve, one would like to determine conditions under which eq-TR on the covering curve is equivalent to ordinary TR on the quotient. In particular, one would seek an explicit relation of the form
\begin{align}
\pi^*\omega^{\mathrm{fold}}_{g,n}= \mathcal {F}\!\left( \omega^{\mathrm{eq}}_{g,n} \right),
\end{align}
where $\mathcal {F}$ is determined by the group action and its characters. 

\item The first instance of such a symmetry was investigated from the meromorphic connection side establishing a reduction from some Painlev\'e transcendents (the higher flows of the Painlev\'e IV hierarchy) to other Painlev\'e transcendents (the Flaschka-Newell Painlev\'e II hierarchy). In this case with $g\geq 1$, the symplectic manifolds of the two hierarchies are identical after the symmetry reduction. The folding on the meromorphic connection side works well in a similar method to the operation studied here. The TR side deals with an equivariant spectral curve of higher genus and it would be interesting to study the folding of such higher genus cases. The hierarchy of Painlev\'e systems provides a particularly natural setting to search for such examples (where the genus of the spectral curves are $\geq 1$ corresponding to a non-trivial symplectic structure), since symmetry reductions already produce many lower-dimensional invariant loci. A successful generalization would establish folding not merely as a useful operation in individual examples, but as a geometric operation acting simultaneously on the moduli of meromorphic connections, their spectral curves, and topological recursion invariants. The Bessel--Clifford example studied here may then be regarded as the simplest non-trivial model of such a theory.

\item A natural higher-genus example of folding from eq-TR
to ordinary TR is provided by the elliptic curve
\begin{align}\label{hyperelliptic}
    y^2=x^4+ax^2+b,
\end{align}
where $b\neq 0$ and $a^2\neq 4b$.
This curve admits the same Klein four-group of
symmetries considered in this article. In particular,
the involution
$
    \iota_2:(x,y)\longmapsto(-x,-y)
$
acts freely on the smooth compactification.
Introducing the invariant coordinates
$
    x_{\mathrm{new}}=x^2,
    y_{\mathrm{new}}=\frac{y}{2x},
$
one obtains the quotient curve
\[
    4x_{\mathrm{new}}y_{\mathrm{new}}^2
    =x_{\mathrm{new}}^2
     +a x_{\mathrm{new}}+b.
\]
The quotient is again an elliptic curve, and the
folding defines an unramified double covering,
equivalently a degree-two isogeny.
Although the genus is preserved in this example,
the period lattices are related nontrivially.
This suggests investigating how folding acts on
homology bases, period matrices, and normalized
Bergman kernels. In particular, it would be
interesting to determine under which conditions
the equivariant topological recursion correlators
on the covering curve descend to ordinary
topological recursion correlators on the quotient,
and how the normalization of the bidifferential
affects this correspondence. More generally,
higher-genus folding need not preserve the genus,
as can already be seen by considering other
involutions of the same curve. A related question is whether the elliptic spectral curve arising
in the spin enumeration of hypermaps in
\cite{ABDKS2026SpinHypermaps} admits a folding interpretation
which might be different to the one considered here. 

\item A different possible application arises in
knot theory. Topological recursion has been used
to compute colored HOMFLY invariants of torus
knots \cite{Brini:2011wi}, starting from an
$SL(2,\mathbb Z)$ transformation of the
resolved-conifold mirror curve. The resulting
description involves fractional winding
contributions, whose unwanted terms can be
removed by Galois averaging.
In \cite{Gu:2014yba}, this construction
was proposed using topological recursion on
augmentation curve, together with a
calibrated annulus kernel. The modification of
the bidifferential is similar but still different to the bidifferential appearing in
equivariant topological recursion \cite{Giacchetto:2021atk}. The conjecture of \cite{Gu:2014yba} about topological recursion on the augmentation curve of torus knots is proved in \cite{ShadrinHock:Augmentation}.
The construction does not arise
from a finite-group action on the spectral
curve itself.
This raises the question of whether the
calibration procedure can also be understood
geometrically as a folding, possibly after
passing to a suitable covering curve.
In particular, one may ask whether the
calibrated annulus kernel can be expressed
as a character-weighted sum of pullbacks
of a bidifferential on a covering curve,
and whether such a correspondence extends
to the stable correlators.

\end{itemize}

\bibliographystyle{halpha-abbrv}
\bibliography{References}

\end{document}